\documentclass[11pt,a4pages]{article}

\usepackage{circularresolution}

\tikzstyle{clause}=[minimum size=20pt,inner sep=2pt]
\tikzstyle{formula}=[minimum size=20pt,inner sep=2pt]
\tikzstyle{rule}  =[circle,draw=black,thick,minimum size=15pt,inner sep=2pt]
\tikzstyle{smallrule}  =[circle,draw=black,thick,minimum size=5pt,inner sep=2pt]
\tikzstyle{formulavertex}  =[rectangle,draw=black,thick,minimum size=18pt,inner sep=1pt]
\tikzstyle{formulavertexnobox}  =[minimum size=18pt,inner sep=1pt]
\tikzstyle{inferencevertex}  =[circle,draw=black,thick,minimum size=18pt,inner sep=1pt]
\tikzstyle{inferencevertexsmall}  =[circle,draw=black,thick,minimum size=8pt,inner sep=1pt]
\tikzstyle{edge} = [draw,thick,>=stealth,->]

\begin{document}

\title{{\bf Circular (Yet Sound) Proofs in Propositional Logic}} 

\author{Albert Atserias \\ Universitat
  Polit\`ecnica de Catalunya \and Massimo Lauria \\ Sapienza -
  Universit\`a di Roma} \date{\today}

\newcommand\blfootnote[1]{%
  \begingroup
  \renewcommand\thefootnote{}\footnote{#1}%
  \addtocounter{footnote}{-1}%
  \endgroup
}

\maketitle

\begin{abstract}
Proofs in propositional logic are typically presented as trees of
derived formulas or, alternatively, as directed acyclic graphs of
derived formulas.
This distinction between tree-like vs.\@ dag-like structure is
particularly relevant when making quantitative considerations
regarding, for example, proof size. Here we analyze a more general
type of structural restriction for proofs in rule-based proof systems.
In this definition, proofs are directed graphs of derived formulas in
which cycles are allowed as long as every formula is derived at least
as many times as it is required as a premise. We call such proofs
``circular''. We show that, for all sets of standard inference rules
with single or multiple conclusions, circular proofs are sound.
We start the study of the proof complexity of circular proofs at
Circular Resolution, the circular version of Resolution.
We immediately see that Circular Resolution is stronger than Dag-like
Resolution since, as we show, the propositional encoding of the
pigeonhole principle has circular Resolution proofs of polynomial
size. Furthermore, for derivations of clauses from clauses, we show
that Circular Resolution is, surprisingly, equivalent to
Sherali-Adams, a proof system for reasoning through polynomial
inequalities that has linear programming at its base.
As corollaries we get: 1) polynomial-time (LP-based) algorithms that
find Circular Resolution proofs of constant width, 2) examples that
separate Circular from Dag-like Resolution, such as the pigeonhole
principle and its variants, and 3) exponentially hard cases for
Circular Resolution. Contrary to the case of Circular Resolution, for
Frege we show that circular proofs can be converted into tree-like
proofs with at most polynomial overhead.
 \end{abstract}

\blfootnote{A preliminary shorter version of this paper was published
  in the Proceedings of 22nd International Conference on Theory and
  Applications of Satisfiability Testing (SAT 2019), Lecture Notes in
  Computer Science 11628, Springer 2019, ISBN 978-3-030-24257-2, pp.
  1-18, Lisbon, Portugal, July 9-12, 2019. The full version of this
  paper has been published in ACM Trans. Comput. Logic. 24(3):1--20.
  This pre-print version can also be found at CoRR abs/1802.05266.}

\section{Introduction}

Logical proofs are traditionally presented as sequences of formulas where
each formula is either a hypothesis, or is deduced from some previous
formulas in the sequence by the means of some inference step.
In rule-based proofs systems each inference step is achieved by
instantiating one in some specific and finite set of inference rules.
Equivalently, any such proof can be represented by a directed acyclic graph, or
\introduceterm{dag}, with one vertex for each formula in the sequence,
and edges pointing forward from the premises to the conclusions of
each inference rule application.

In this paper we discuss an alternative and more general way of
composing proofs: we allow cycles in the graph.
In general, and not suprisingly, unlimited circular reasoning of this
type may be unsound. However, when every formula is derived at least
as many times as it is used as a premise of an inference step, we show
that soundness is guaranteed.
Hence we appropriately call these objects with the name of
\emph{circular proofs}.

More formally, our soundness requirement is phrased in terms of
\introduceterm{flow assignments}: each rule application must carry a
\introduceterm{flow}, a positive integer which intuitively means that
in order to produce that many copies of each conclusion of the rule we
must have produced at least that many copies of each premise.
Flow assignments induce a notion of \introduceterm{balance} of
a formula in the proof, which is the difference between the number of
times that the formula is produced as a conclusion and the number of
times that it is required as a premise.
Given these definitions, a proof-graph will be an actual circular
proof if it admits a flow assignment that satisfies the following
\emph{flow-balance} condition: the only formulas of strictly negative
balance are the hypotheses, and the goal formulas display strictly
positive balance.
With this interpretation of flows,
circular proofs have the appealing flavour of a network in which
demands are fulfilled by the hypotheses, and flow towards the
goal formulas, which produce surplus. Accordingly, and in analogy with
the theory of classical network flows
\cite{Schrijver2003Combinatorial}, it makes no difference whether the
flows are required to be integers or real numbers, and valid flow
assignments can be found efficiently, when they exist, by linear
programming techniques.

While proof-graphs with unrestricted cycles are, in general, unsound,
we show that circular proofs \emph{are} sound.
We prove this in two ways.  The first one is combinatorial in
nature and is phrased in the style of traditional soundness proofs in
standard proof systems. Concretely, given a truth assignment that
falsifies a goal formula, the soundness proof constructs a path of
falsified formulas until it reaches a hypothesis, and does so by
induction on the total flow-sum of the flow assignment that satisfies
the flow-balance condition. 
This proof is more informative and intuitive, but it is also
inefficient in the sense that the process of building the path is not
polynomial in the size of the proof.
The second proof is (semi-)algebraic and
is phrased in the style of the duality theorem for linear
programming. Concretely, we phrase the unsoundness of the proof as the
feasibility of a linear program and observe that the existence of a
flow assignment that satisfies the flow-balance condition gives a
witness of its infeasibility.
This second proof is less intuitive but can be efficiently simulated by
non-circular argument in relatively strong proof systems. It will be
useful when trying to understand the power of circular reasoning. 

\paragraph{Proof complexity of circular proofs} With all the
definitions in place, we proceed to the study of the power of circular
proofs from the perspective of propositional proof complexity.

For Resolution, we show that circularity \emph{does} make a real
difference. First we show that the standard propositional formulation
of the pigeonhole principle has Circular Resolution proofs of
polynomial size. This is in sharp contrast with the well-known fact
that Resolution \emph{cannot count}, and that the pigeonhole principle
is exponentially hard for (tree-like and dag-like)
Resolution~\cite{Haken1985}. Second we observe that the LP-based proof
of soundness of Circular Resolution can be formalized in the
Sherali-Adams proof system (with twin variables), which is a proof
system for reasoning with polynomial inequalities that has linear
programming at its base~\cite{SheraliAdams1990}.  Sherali-Adams was
originally conceived as a hierarchy of linear programming relaxations
for integer programs, but it has also been studied from the
perspective of proof complexity in recent years~\cite{Dantchev2007,
  DantchevMartinRhodes2009,
  SegerlindPitassi2012,AtseriasLauriaNordstrom2016}.

Surprisingly, it turns out that the converse simulation is also true!
For deriving clauses from clauses, Sherali-Adams proofs translate
efficiently into Circular Resolution proofs. Moreover, both
translations, the one from Circular Resolution into Sherali-Adams and
its converse, are efficient in terms of their natural parameters:
length/size and width/degree. As corollaries we obtain for Circular
Resolution all the proof complexity-theoretic properties that are
known to hold for Sherali-Adams: 1) a polynomial-time (LP-based) proof
search algorithm for proofs of bounded width, 2) length-width
relationships, 3) separations from dag-like length and width, and 4)
explicit exponentially hard examples.

Going beyond resolution we address the question of how circularity
affects more powerful propositional proof systems.  For Frege systems,
which operate with arbitrary propositional formulas through the
standard textbook inference rules, we show that circularity adds no
power: the circular, dag-like and tree-like variants of Frege
polynomially simulate one another.  The equivalence between the
dag-like and tree-like variants of Frege is well-known
\cite{Kra94BoundedArithmetic}; here we add the circular variant to the
list.  We prove this by formalizing the LP-based proof of soundness
for Circular Frege within Tree-like Frege itself. To achieve this we
make strong use of the formalization of linear arithmetic in Frege
that was developed by Buss in order to get efficient Frege proofs of
the pigeonhole principle~\cite{Buss1987}, and that was developed
further by Goerdt to show that Tree-like Frege simulates the
Cutting Planes proof system~\cite{Goerdt1991CPvsFrege}.

\paragraph{Earlier work} 
The idea of allowing cyclic, circular or non-wellfounded proofs has
been studied by several communities since at least 20 years ago, from
modal $\mu$-calculus~\cite{NiwinskiWalukiewicz1996}, to predicate
logic with inductive definitions~\cite{Brotherston2006Thesis,
  BrotherstonSimpson2010}, fragments of arithmetic~\cite{Simpson2017,
  Das2019}, provability logics~\cite{Shamkanov2020}, and linear
logic~\cite{Fortier2014Thesis}. For classical propositional logic
proper and in the context of proof complexity, we are not aware of any
work on cyclic, circular, or non-wellfounded proofs that appeared
earlier than the conference version of this
paper~\cite{AtseriasLauria2019}. It seems that our flow-based
definition of circular proofs had not been considered before.

Niwi\'nksi and Walukiewicz \cite{NiwinskiWalukiewicz1996} introduced
an infinitary tableau method for the modal $\mu$-calculus. The proofs
are regular infinite trees that are represented by finite graphs with
cycles, along with a decidable \emph{progress condition} on the cycles
to guarantees their soundness. A sequent calculus version of this
tableau method was proposed in \cite{DaxHofmannLange2006}, and
explored further in \cite{Studer2008}. In his PhD thesis, Brotherston
\cite{Brotherston2006Thesis} introduced a \emph{cyclic} proof system
for the extension of first-order logic with inductive definitions; see
also \cite{BrotherstonSimpson2010} for a journal article presentation
of the results. The proofs in~\cite{BrotherstonSimpson2010} are
ordinary proofs of the first-order sequent calculus extended with the
rules that define the inductive predicates, along with a set of
\emph{backedges} that link equal formulas in the proof. The soundness
is guaranteed by an additional \emph{infinite descent condition} along
the cycles that is very much inspired by the progress condition in
Niwi\'nski-Walukiewicz' tableau method. We refer the reader to
Section~8 from~\cite{BrotherstonSimpson2010} for a careful overview of
the various flavours of proofs with cycles for logics with inductive
definitions.

Shoesmith and Smiley \cite{ShoesmithSmiley1978MultipleConclusion}
initiate the study of inference based propositional proofs with
multiple conclusions. In order to do so they introduce a graphical
representation of proofs where nodes represents either formulas or
inference steps, in a way similar to our definition in
Section~\ref{sec:prelim}. While most of that book does not consider
proof with cycles, in Section 10.5 they do mention briefly this
possibility but they do not analyze it any further.

The Sherali-Adams hierarchy of linear programming relaxations has
received considerable attention in recent years for its relevance to
combinatorial optimization and approximation algorithms; see the
original \cite{SheraliAdams1990}, and \cite{AuTuncel2016} for a recent
survey. In its original presentation, the Sherali-Adams hierarchy can
already be thought of as a proof system for reasoning with polynomial
inequalities, with the levels of the hierarchy corresponding to the
degrees of the polynomials. For propositional logic, the system was
studied in \cite{Dantchev2007}, and developed further in
\cite{SegerlindPitassi2012,AtseriasLauriaNordstrom2016}. Those works consider
the version of the proof system in which each propositional variable
$X$ comes with a formal \emph{twin variable} $\bar{X}$, that is to be
interpreted by the negation of $X$. This is the version of
Sherali-Adams that we use. It was already known from
\cite{DantchevMartinRhodes2009} that this version of the Sherali-Adams proof
system polynomially simulates standard Resolution, and has
polynomial-size proofs of the pigeonhole principle.

\section{Preliminaries} \label{sec:prelim}

\subsection{Formulas}

A \emph{literal} is a variable $X$ or the negation of a variable
$\overline{X}$; we also say that literal $\overline{X}$ is the
negation of literal $X$, and vice-versa.
The class of formulas in \emph{negation normal form} is the smallest
class of formulas that contains the literals and is closed under
conjunction~$\wedge$ and disjunction~$\vee$.  If~$A$ is a formula in
negation normal form, we write~$\overline{A}$ for its dual formula,
which is defined recursively as follows: If~$A$ is a literal,
then~$\overline{A}$ is its negation.  If~$A = B \vee C$,
then~$\overline{A} = \overline{B} \wedge \overline{C}$.
If~$A = B \wedge C$,
then~$\overline{A} = \overline{B} \vee \overline{C}$. Note that the
dual of the dual of~$A$ is~$A$ itself. A \emph{truth-assignment} is a
mapping that assigns a truth-value \emph{true}~($1$) or
\emph{false}~($0$) to each variable. Truth-assignments evaluate
formulas in the natural way through the standard interpretations of
negation, conjunction, and disjunction.  The \emph{empty formula} is
denoted by~$\emptyformula$, and is always false by convention. Its
complement~$\overline{\emptyformula}$ is denoted by~$\fullformula$ and
is always true by convention.  If a truth-assignment evaluates a
formula to true, then we say that is satisfies it. A
\emph{substitution} is a mapping that assigns a formula to each
variable. Applying a substitution to a formula means replacing all
variables by the formulas to which they are mapped to by the
substitution, simultaneously all at once.

We think of disjunction as binding \emph{sets} of formulas, or,
equivalently, as a binary operation on formulas that is associative,
commutative and idempotent. This means that the
formula~$(A \vee B) \vee C$ is considered the same
as~$A \vee (B \vee C)$, which we just write as~$A \vee B \vee C$. Also
the formula~$A \vee B$ is considered the same as~$B \vee A$, and the
formula~$A \vee A$ is considered the same as~$A$. Similarly, we view
conjunction as a binary operation on formulas that is associative,
commutative and idempotent.  The empty formula~$\emptyformula$ and its
complement~$\fullformula$ are the neutral elements of~$\vee$
and~$\wedge$, respectively.  Thus the formulas~$\emptyformula \vee A$
and~$\fullformula \wedge A$ are considered the same as~$A$. These
conventions about disjunctions and conjunctions mean that our syntax
for formulas in negation normal form could have been defined as
follows:~(1)~every literal is a formula,~(2)~if~$S$ is a set of
formulas none of which starts with~$\bigvee$, then~$\bigvee\!S$ is a
formula,~(3) if~$S$ is a set of formulas none of which starts
with~$\bigwedge$, then~$\bigwedge\!S$ is a formula, and~(4)~nothing
else is a formula. The empty formula~$\emptyformula$ and its
complement~$\fullformula$ are taken to be~$\bigvee\!\emptyset$
and~$\bigwedge\!\emptyset$, respectively.  We adopt this
\emph{unbounded fan-in} definition of syntax, but continue to use the
notation~$A_1 \vee \cdots \vee A_n$ even if the~$A_i$ may be
disjunctions themselves. The \emph{size}~$s(A)$ of a formula~$A$ is
defined inductively: if~$A$ is a literal, then~$s(A) = 1$, and
if~$A = \bigvee\!S$ or~$A = \bigwedge\!S$,
then~$s(A) = 1+\sum_{B \in S} s(B)$.

An \emph{elementary tautology}
is a formula of the form~$\overline{A} \vee A$, where~$A$ is a
formula. Note that by the definition of the dual of a formula (and the
convention to read disjunctions up to associativity), the
formula~$\overline{A} \vee \overline{B} \vee (A \wedge B)$ is an
elementary tautology.  If~$\Gamma$ is a set of formulas, a disjunction
of formulas in~$\Gamma$ is a formula of the
form~$A_1 \vee \cdots \vee A_m$, where~$m$ is a non-negative integer
and each~$A_i$ is a formula in~$\Gamma$. Disjunctions of formulas
in~$\Gamma$ are also called~$\Gamma$-clauses or~$\Gamma$-cedents. A
clause is a disjunction of literals.

\subsection{Inference-Based Proofs}

An \emph{inference rule} is given by a sequence of \emph{premise
  formulas}~$A_1,\ldots,A_r$ and a sequence of \emph{conclusion
  formulas}~$B_1,\ldots,B_s$ with the property that every truth
assignment that satisfies all the premises also satisfies all the
conclusions. Here are four important examples:
\begin{equation} 
\frac{}{A \vee \overline{A}} \;\;\;\;\;\;\;\; \frac{C
\vee A \;\;\;\;\;\;\;\; D \vee \overline{A}}{C \vee D}
\;\;\;\;\;\;\;\; \frac{C \vee A \;\;\;\;\;\;\;\; D \vee B}{C \vee D
\vee (A \wedge B)} \;\;\;\;\;\;\;\; \frac{C}{C \vee D}.
\label{eqn:Fregerules}
\end{equation}
These are the standard inference rules of a Tait-style calculus for
propositional logic \cite{Tait1968}. The rules are called
\emph{axiom}, \emph{cut}, \emph{introduction of conjunction}, and
\emph{weakening}, respectively. An \emph{instance} of an inference
rule is obtained from applying a substitution to its variables. Note
that every instance of a rule is a rule itself, which has its own
premise formulas and conclusion formulas.
 
In almost all classical examples in the literature, inference rules
have a single conclusion formula. The reason for this is that for
classical (i.e., non-circular) proofs one may simply split a rule
with~$s$ conclusion formulas into~$s$ different single-conclusion
rules, with little conceptual change. However, for circular proofs a
specific rule with two conclusion formulas will play an important
role; this is the \emph{symmetric weakening}, or \emph{split}, rule:
\begin{equation}
\frac{C}{C \vee A \;\;\;\;\;\;\;\; C \vee \overline{A}}.
\label{eqn:splitrule} 
\end{equation} 
When we apply~\eqref{eqn:splitrule} we say that we split~$C$ on~$A$.
In all these examples the formulas~$C$,~$D$, and~$A$ could be a single
literal, the empty formula~$\emptyformula$, or its
complement~$\fullformula$.

Fix a set~$\mathscr{R}$ of inference rules, a set~$A_1,\ldots,A_m$ of
\emph{hypothesis formulas}, and a \emph{goal formula}~$A$. A
\emph{proof of~$A$ from~$A_1,\ldots,A_m$}, also called a
\emph{derivation}, is a finite sequence of formulas that ends in~$A$
and such that each formula in the sequence is either contained
in~$A_1,\ldots,A_m$, or is one of the conclusion formulas of an
instance of an inference rule in~$\mathscr{R}$ that has all its
premise formulas appearing earlier in the sequence. A derivation
of~$A$ \emph{from nothing} is also called a \emph{proof of~$A$}. A
\emph{refutation of~$A_1,\ldots,A_m$} is a derivation of the empty
formula~$\emptyformula$ from~$A_1,\ldots,A_m$. The \emph{length} of
the derivation is the length of the sequence, and its \emph{size} is
the sum of the sizes of the formulas in the sequence.

Proofs and derivations are naturally represented through directed
acyclic graphs, a.k.a.\ \emph{dags}; see
Figure~\ref{fig:dagproof}. The graph has one \emph{formula-vertex} for
each formula in the sequence, and one \emph{inference-vertex} for each
inference step that produces a formula in the sequence. Each
formula-vertex is labelled by the corresponding formula, and each
inference-vertex is labelled by the corresponding instance of the
corresponding inference rule. Each inference-vertex that is labelled
by an inference rule that has~$r$ premise formulas and~$s$ conclusion
formulas has, accordingly,~$r$ incoming edges from the corresponding
premise formula-vertices, and at least one and at most~$s$ outgoing
edges towards the corresponding conclusion formula-vertices. The
directed acyclic graph of a proof~$\Pi$ is its \emph{proof-graph}, and
is denoted by~$G(\Pi)$. A proof~$\Pi$ is called \emph{tree-like}
if~$G(\Pi)$ is a tree.

\begin{figure}
\begin{center}
\begin{tikzpicture} \draw[help lines, white, use as bounding box] (-3,-0.5) grid (8,2.5);
\node[formulavertex] (a1) at (-2,2) {$A_1$};
\node[formulavertex] (a2) at (-2,1) {$A_2$};
\node[inferencevertex] (r1) at (-0.5,1.5) {$R_1$};
\node[inferencevertex] (r2) at (-0.5,0) {$R_2$};
\node[formulavertex] (a3) at (1,1) {$A_3$};
\node[inferencevertex] (r3) at (2.5,0.5) {$R_3$};
\node[formulavertex] (a4) at (1,0) {$A_4$};
\node[formulavertex] (a5) at (4,0.5) {$A_5$};
\node[inferencevertex] (r4) at (2.5,1.5) {$R_4$};
\node[formulavertex] (a6) at (4,1.5) {$A_6$};
\node[inferencevertex] (r5) at (5.5,1) {$R_5$};
\node[formulavertex] (a7) at (7,1) {$A_7$};
\draw[edge] (a1) to (r1);
\draw[edge] (a2) to (r1);
\draw[edge] (r1) to (a3);
\draw[edge] (r2) to (a4);
\draw[edge] (a4) to (r3);
\draw[edge] (a3) to (r4);
\draw[edge] (a3) to (r3);
\draw[edge] (r3) to (a5);
\draw[edge] (r4) to (a5);
\draw[edge] (r4) to (a6);
\draw[edge] (a6) to (r5);
\draw[edge] (a5) to (r5);
\draw[edge] (r5) to (a7);
\end{tikzpicture}
\end{center}
\caption{The directed acyclic graph representation of a proof of $A_7$
  from the set of hypothesis formulas $A_1$ and $A_2$ through the
  inference rules $R_1,\ldots,R_5$. Formula-vertices are represented
  by boxes and inference-vertices are represented by circles. Formula
  $A_3$ is used twice as the premise of an inference, and $A_5$ is
  produced twice as the conclusion of an inference. All rules except
  $R_4$ have exactly one conclusion formula; $R_4$ has two. All rules
  except $R_2$ have at least one premise formula; $R_2$ has none.}
\label{fig:dagproof}
\end{figure}

\subsection{Frege and Resolution Proof Systems}

An inference-based proof system is given by a set of allowed inference
rules, a set of allowed formulas, and a set of allowed
proof-graphs. Two typical sets of allowed proof-graphs are the set of
dags, for \emph{dag-like} proofs, and the set of trees, for
\emph{tree-like} proofs. If the set of allowed proof-graphs is
omitted, dag-like is assumed by default. A proof system~$P$ is said to
\emph{polynomial simulate} another proof system~$P'$ if there is a
polynomial-time algorithm that, given a proof~$\Pi'$ in~$P'$ as input,
computes a proof~$\Pi$ in~$P$, such that~$\Pi$ has the same goal
formula and the same hypothesis formulas as~$\Pi'$.  \emph{Frege} and
\emph{Resolution} are both inference-based proof systems, as defined
next.

In our definition of Frege the set of allowed inference rules are
axiom, cut, introduction of conjunction, and weakening as defined
in~\eqref{eqn:Fregerules}, and the set of allowed formulas is the set
of all formulas in negation normal form. Being equivalent to a
Tait-style calculus, our definition of Frege is sound and
(implicationally) complete for formulas in negation normal form. This
means that if~$A$ has a Frege proof from the set of hypothesis
formulas~$A_1,\ldots,A_m$, then every truth assignment that satisfies
all the formulas in~$A_1,\ldots,A_m$ also satisfies~$A$, and
vice-versa.

In our definition of Resolution the only allowed inference rule is cut
and the allowed formulas are the clauses. This proof system is sound
and complete as a refutation system. This means that if the set of
clauses $A_1,\ldots,A_m$ has a Resolution refutation, then there is no
truth-assignment that satisfies all clauses $A_1,\ldots,A_m$
simultaneously, and vice-versa. In order to turn Resolution into a
proof system that is sound and complete for deriving clauses from
clauses, one needs to add the axiom and weakening rules to the set of
allowed rules. The \emph{width} of a Resolution proof is the number of
literals of its largest clause.

\subsection{Frege and Resolution with Symmetric Rules}
\label{sec:symmetricrules}

Consider an inference-based proof system in which elementary
tautologies of the form~$A \vee \overline{A}$ may be introduced at any
point in the proof through the axiom rule, and that in addition has
the following two nicely symmetric-looking inference rules:
\begin{equation}
\frac{C \vee A \;\;\;\;\;\;\;\; C \vee \overline{A}}{C}
\;\;\;\;\;\;\;\;\;\;\;\;
\frac{C}{C \vee A \;\;\;\;\;\;\;\; C \vee \overline{A}}.
\label{eqn:homogeneousFrege} 
\end{equation}
These rules are called \emph{symmetric cut} and \emph{symmetric
  weakening}, or \emph{split}, respectively. Note the subtle
difference between the symmetric cut rule and the standard cut rule
in~\eqref{eqn:Fregerules}: in the symmetric cut rule, both premise
formulas have the same \emph{side formula}~$C$. This difference is
minor: an application of the non-symmetric cut rule that
derives~$C \vee D$ from~$C \vee A$ and~$D \vee \overline{A}$ may be
efficiently simulated as follows (here and in what follows, the
applicability of the rules has to be read up to associativity,
symmetry, and idempotency of disjunctions and conjunctions, and the
second conclusion of the split rule has been suppressed from the list
of derived formulas whenever it is not needed):
\begin{center}
\begin{tabular}{llll}
1. & $C \vee A \vee D$ & \;\;\; & by split on $C \vee A$, \\
2. & $D \vee \overline{A} \vee C$ & & by
split on $D \vee \overline{A}$, \\
3. & $C \vee D$ & & by symmetric cut on 1 and 2.
\end{tabular}
\end{center}
Note also that the
rules in~\eqref{eqn:homogeneousFrege} do not include a rule for
\emph{introduction of conjunction} as in~\eqref{eqn:Fregerules}. In
the presence of the elementary tautologies (or, equivalently, the
axiom rule), this difference is again minor: an application of the
introduction of conjunction rule that derives $C \vee D \vee (A \wedge
B)$ from $C \vee A$ and $D \vee B$ may be efficiently simulated by the
following sequence:
\begin{center}
\begin{tabular}{llll}
1. & $\overline{A} \vee \overline{B} \vee (A \wedge B)$ & \;\;\;& as an elementary tautology, \\
2. & $\overline{A} \vee \overline{B} \vee (A \wedge B) \vee C$ & & by split on 1, \\
3. & $C \vee A \vee \overline{B} \vee (A \wedge B)$ & & by split on $C \vee A$, \\
4. & $C \vee \overline{B} \vee (A \wedge B)$ & & by symmetric cut on 2 and 3, \\
5. & $C \vee \overline{B} \vee (A \wedge B) \vee D$ & & by split on 4, \\
6. & $D \vee B \vee C \vee (A \wedge B)$ & & by split on $D \vee B$, \\
7. & $C \vee D \vee (A \wedge B)$ & & by symmetric cut on 5 and 6.
\end{tabular}
\end{center} 
Thus, for all practical purposes, the Frege proof system as defined in
the previous section and the proof system defined here are equivalent.
The same observation applies to Resolution. In this case the
elementary tautologies are of the form $X \vee \overline{X}$, where
$X$ is a variable, and the instances of the symmetric cut and split
rules in~\eqref{eqn:homogeneousFrege} have a variable for its
\emph{cut formula} $A$.
Note that an application of the standard weakening rule that derives
the clause $C \vee D$ from the clause $C$ may be efficiently simulated
by $|D|$ many applications of the split rule by introducing one
literal at a time; here $|D|$ denotes the number of literals in $D$.

\subsection{Sherali-Adams Proof System}

Let $X_1,\ldots,X_n$ be variables that are intended to range over
$\{0,1\}$, and let $\bar{X}_1,\ldots,\bar{X}_n$ be their \emph{twins},
with the intended meaning that $\bar{X}_i = 1-X_i$.  Let
$A_1,\ldots,A_m$ and $A$ be polynomials on the variables
$X_1,\ldots,X_n$ and $\bar{X}_1,\ldots,\bar{X}_n$.  A
\emph{Sherali-Adams proof of $A \geq 0$ from $A_1 \geq 0,\ldots,A_m
  \geq 0$} is a polynomial identity of the form
\begin{equation}
\sum_{j=1}^t Q_j P_j = A, \label{eqn:sheraliadamsproof}
\end{equation}
where each $Q_j$ is a non-negative linear combination of monomials on
the variables $X_1,\ldots,X_n$ and $\bar{X}_1,\ldots,\bar{X}_n$, and
each $P_j$ is a polynomial among $A_1,\ldots,A_m$ or one among the
following set of \emph{basic} polynomials:
\begin{equation}
\begin{array}{lllll}
X_i-X_i^2, & 1-X_i-\bar{X}_i, &
X_i^2-X_i, & X_i+\bar{X}_i-1, & 1. 
\end{array} 
\label{eqn:allowedindefinition}
\end{equation}
Observe that all basic polynomials, as well as all~$Q_{j}$'s, being
non-negative linear combinations of monomials, are non-negative
on~$\{0,1\}$. If~$A_1,\ldots,A_m$ are also non-negative, then,
by~\eqref{eqn:sheraliadamsproof},~$A$ must be non-negative. It follows
that the proof system is sound. By Theorem~3 in the original paper by
Sherali and Adams~\cite{SheraliAdams1990}, the proof system is
complete for deriving linear inequalities from linear
inequalities. Therefore, when clauses are encoded by linear
inequalities in the natural way, the proof system is also complete
(see also Lemma~4.2 in~\cite{AtseriasLauriaNordstrom2016} and also
Section~\ref{sec:resolution} in this paper).
In a Sherali-Adams proof each~$Q_{j}$ is given explicitly as a
positive linear combination of monomials, where the coefficients in
the linear combination are rational numbers written in binary. It
follows that the identity asserted by
equation~\eqref{eqn:sheraliadamsproof} can be checked in polynomial
time with respect to the length of the proof itself. These three facts
together imply that Sherali-Adams is a Cook-Reckhow proof system.

The \emph{degree} of the proof is the maximum of the degrees of the
polynomials~$Q_jP_j$ in~\eqref{eqn:sheraliadamsproof}.  The
\emph{monomial size} of the proof is the sum of the monomial sizes of
the polynomials~$Q_jP_j$ in~\eqref{eqn:sheraliadamsproof}, where the
monomial size of a polynomial is the number of monomials with non-zero
coefficient in its unique representation as a linear combination of
monomials. The \emph{bit size} of the proof is the sum of the bit
sizes of the polynomials~$Q_jP_j$ in~\eqref{eqn:sheraliadamsproof},
where the bit size of a polynomial is the sum of the bit sizes of its
terms, where the bit size of a term is the number of bits that it
takes to describe the monomial and to write the rational coefficient
in binary.

\section{Circular Proofs} \label{sec:definition}

Informally, a circular proof will be defined as a ``proof with
cycles''. Formally such objects will be called \emph{circular
  pre-proofs} because, in general, they are not sound. We define
\emph{circular proofs} by adding a global yet efficiently checkable
requirement on the definition of pre-proof that guarantees its
soundness.  

\subsection{Definition}

A \emph{circular pre-proof} is a directed graph with two types of
vertices: formula-vertices and inference-vertices. All edges of the
graph go from a formula-vertex to an inference-vertex, or from an
inference-vertex to a formula-vertex. Thus, the graph is
\emph{bipartite}. Each formula-vertex is labelled by a formula, and
each inference-vertex is labelled by an instance of an inference rule
that has the formulas that label its in-neighbors as premises, and the
formulas that label its out-neighbors as conclusions. If~$\Pi$ is a
pre-proof, we use~$G(\Pi)$ to denote the underlying bipartite graph,
ignoring the labels. When~$\Pi$ is clear from the context, we
write~$I$ and~$J$ for the sets of inference- and formula-vertices
of~$G(\Pi)$, respectively, and~$N^-(u)$ and~$N^+(u)$, respectively,
for the sets of in- and out-neighbours of a vertex~$u$
of~$G(\Pi)$. Figure~\ref{fig:circularpreproof} illustrates these
definitions.

\begin{figure}
\begin{center}
\begin{tikzpicture} \draw[help lines, white, use as bounding box] (-3,-0.5) grid (5,2);
\node[formulavertex] (a1) at (-2,0.5) {$A_1$};
\node[inferencevertex] (r1) at (-0.5,1) {$R_1$};
\node[inferencevertex] (r2) at (-0.5,0) {$R_2$};
\node[formulavertex] (b1bis) at (1,1) {$B_1$};
\node[inferencevertex] (r3) at (2.5,0.5) {$R_3$};
\node[formulavertex] (a2) at (1,0) {$A_2$};
\node[formulavertex] (a3) at (4,0.5) {$A_3$};
\draw[edge] (a1) to (r1);
\draw[edge] (r1) to (b1bis);
\draw[edge] (r2) to (a2);
\draw[edge] (a2) to (r3);
\draw[edge] (b1bis) to (r3);
\draw[edge] (r3) to (a3);
\draw[thick] plot [smooth] coordinates { (1.3,1.1) (1.6,1.7) (-1.1,1.7) (-0.95,1.1) };
\draw[thick,-stealth] (-0.95,1.11) -- (-0.85,1.1);
\end{tikzpicture}
\caption{A circular pre-proof. The vertex labelled~$R_1$ has two
  premises,~$A_1$ and~$B_1$, and one conclusion,~$B_1$ itself. The
  vertex labelled~$R_3$ has two premises,~$B_1$ and~$A_2$, and one
  conclusion,~$A_3$. The vertex labelled~$R_2$ has no premises and one
  conclusion,~$A_2$.}
\label{fig:circularpreproof}
\end{center}
\end{figure}

By the correspondence between proofs and their dags as defined in
Section~\ref{sec:prelim}, a circular pre-proof whose underlying graph
is a dag is just the same as an ordinary proof. However, general
circular pre-proofs need not be sound; see
Figure~\ref{fig:circularunsound} for an example of an unsound circular
pre-proof. In order to ensure soundness we need to require a global
condition~as~defined~next.

\begin{figure}
\begin{center}
\begin{tikzpicture} \draw[help lines, white, use as bounding box] (-3,-0.5) grid (5.5,2.5);
\node[inferencevertex] (r1) at (-2.5,1) {\tiny{AX}};
\node[formulavertex] (a1) at (-1,1) {$X \vee \overline{X}$};
\node[inferencevertex] (r2) at (0.5,1.8) {\tiny{CUT}};
\node[inferencevertex] (r3) at (0.5,0.2) {\tiny{CUT}};
\node[formulavertex] (a2) at (2,1.5) {$X$};
\node[formulavertex] (a3) at (2,0.5) {$\overline{X}$};
\node[inferencevertex] (r4) at (3.5,1) {\tiny{CUT}};
\node[formulavertex] (a4) at (5,1) {$0$};
\draw[edge] (r1) to (a1);
\draw[edge] (a1) to (r2);
\draw[edge] (a1) to (r3);
\draw[edge] (r2) to (a2);
\draw[edge] (r3) to (a3);
\draw[edge] (a2) to (r4);
\draw[edge] (a3) to (r4);
\draw[edge] (r4) to (a4);
\draw[thick] plot [smooth] coordinates { (2.3,1.7) (2.5,2.2) (-0.2,2.3) (0.1,1.9) };
\draw[thick,-stealth] (0.1,1.91) -- (0.2,1.9);
\draw[thick] plot [smooth] coordinates { (2.3,0.3) (2.5,-0.2) (-0.2,-0.3) (0.1,0.1) };
\draw[thick,-stealth] (0.1,0.09) -- (0.2,0.1);
\end{tikzpicture}
\end{center}
\caption{An unsound circular pre-proof: the false empty
  formula~$\emptyformula$ is derived from no hypotheses. We note that
  if we were to assign positive weights to the inference-vertices,
  then it would always be the case that the sum of the weights that
  enter~$X$ minus the sum of the weights that leave~$X$ would always
  be negative (and the same for~$\bar{X}$). As we will see, this turns
  out to be the \emph{only reason} for it not being sound.}
\label{fig:circularunsound}
\end{figure}

A \emph{flow assignment} for a circular pre-proof~$\Pi$ is an
assignment~$F : I \rightarrow \mathbb{R}^+$ of positive real weights,
or \emph{flows}, where~$I$ is the set of inference-vertices of the
graph~$G(\Pi)$ of~$\Pi$. The flow-extended graph that labels each
inference-vertex~$w$ of~$G(\Pi)$ by its flow~$F(w)$ is
denoted by~$G(\Pi,F)$. The \emph{inflow} of a formula-vertex
in~$G(\Pi,F)$ is the sum of the flows of its in-neighbours. Similarly,
the \emph{outflow} of a formula-vertex in~$G(\Pi,F)$ is the sum of the
flows of its out-neighbours. The \emph{balance} of a
formula-vertex~$u$ of~$G(\Pi,F)$ is the inflow minus the outflow
of~$u$, and is denoted by~$B(u)$. In symbols,
\begin{equation}
B(u) := \sum_{w \in N^-(u)} F(w) - \sum_{w \in N^+(u)} F(w).
\label{eqn:balance}
\end{equation}
This notion allows us to define \emph{sources} and \emph{sinks}
in~$G(\Pi,F)$.
The formula-vertices of strictly negative balance are the sources
of~$G(\Pi,F)$, and those of strictly positive balance are the sinks
of~$G(\Pi,F)$.
We think of flow assignments as witnessing a proof of a formula that
labels a sink from the set of formulas that label the
sources. Concretely, for a given set of hypothesis
formulas~$\mathscr{H}$ and a given goal formula~$A$, we say that the
flow assignment \emph{witnesses a proof of~$A$ from~$\mathscr{H}$} if
every source of~$G(\Pi,F)$ is labelled by a formula in~$\mathscr{H}$,
and some sink of~$G(\Pi,F)$ is labelled by the formula~$A$.

Finally, a \emph{circular proof of~$A$ from~$\mathscr{H}$} is a
circular pre-proof for which there exists a flow assignment that
witnesses a proof of~$A$ from~$\mathscr{H}$. The \emph{length} of a
circular proof~$\Pi$ is the number of vertices of~$G(\Pi)$, and the
\emph{size} of~$\Pi$ is the sum of the sizes of the formulas that
label its formula-vertices. Note that this definition of size does not
depend on the weights that witness the proof. As we will see in the
next section, such weights may be assumed to be integral and have
small bit-complexity.

\subsection{Checking the Global Condition}

We still need to argue two facts about circular proofs: 1) that the
existence of a witnessing flow assignment guarantees soundness, and 2)
that its existence can be checked algorithmically in an efficient
way. Soundness is proved in the next section. Here we argue that its
existence can be checked efficiently. One way to do this is by
solving a linear program.

\begin{lemma} \label{lem:checking} There is a polynomial-time
  algorithm that, given as input a circular pre-proof~$\Pi$, a finite
  set of hypothesis formulas~$\mathscr{H}$, and a goal formula~$A$,
  returns a flow assignment for~$\Pi$ that witnesses a proof of~$A$
  from~$\mathscr{H}$, if it exists.
\end{lemma}

\begin{proof}
  Let~$V = I \cup J$ be the set of vertices of the graph~$G(\Pi)$
  of~$\Pi$, partitioned into the set~$I$ of inference-vertices and the
  set~$J$ of formula-vertices.
  Observe that~$N^-(u) \subseteq I$ and~$N^+(u) \subseteq I$ for
  each~$u \in J$. Let~$H \subseteq J$ be the set of formula-vertices
  whose labels are in~$\mathscr{H}$, and let~$a \in J$ be a
  formula-vertex whose label is~$A$ and gets positive balance under
  some flow assignment for~$\Pi$. For each~$w$ in~$I$, let~$Y_w$
  denote a real-valued variable and consider the following instance of
  the linear programming feasibility problem:
\begin{equation*}
(P): \;\;
\begin{array}{lll}
\sum_{w \in N^-(u)} Y_w - \sum_{w \in N^+(u)} Y_w \geq 1 & & \text{ for $u = a$, } \\
\sum_{w \in N^-(u)} Y_w - \sum_{w \in N^+(u)} Y_w \geq 0 & \;\;\;\;\; & \text{ for each $u \in J\setminus (H \cup \{a\})$, } \\
Y_w \geq 1 & & \text{ for each $w \in I$. }
\end{array}
\end{equation*}
We claim that~$(P)$ has a feasible solution~$(y_w)_{w \in I}$ if and
only if there exists a flow
assignment~$F : I \rightarrow \mathbb{R}^+$ that witnesses a proof
of~$A$ from~$\mathscr{H}$ by making the balance~$B(a)$ of~$a$
positive. For the \emph{only if} direction,
define~$F : I \rightarrow \mathbb{R}^+$ by~$F(w) := y_w$, and read-off
the required conditions for~$F$ from the inequalities that
define~$(P)$. For the \emph{if} direction, define~$y_w := F(w)/D$,
where~$D$ is the minimum in the finite
set~$\{ F(w) : w \in I \} \cup \{ B(a) \}$ and~$B(a)$ denotes the
balance of~$a$.  Observe that~$D$ is strictly positive by definition
and the choice of~$a$. The inequalities of~$(P)$ are satisfied
by~$(y_w)_{w \in I}$ also by definition, and by the choice
of~$D$. Since the linear programming feasibility problem can be solved
in polynomial time in the size of the input, the lemma follows.
\end{proof}

By elementary facts about linear programming
(see~\cite{Schrijver2003Combinatorial}), it follows from this proof
that if there is a flow assignment that witnesses a proof, then there
is one with flows that are rational numbers whose bit-complexity is at
most polynomial in the length of the circular pre-proof. By taking
common denominators and multiplying through, the flows can even be
taken to be positive integers of bit-complexity still polynomial in
the length of the pre-proof. We collect these observations in a lemma.

\begin{lemma} \label{lem:integralflow} Let~$\Pi$ be a circular
  pre-proof of length~$\ell$. For every flow assignment~$F$ for~$\Pi$
  there exists another flow assignment~$F'$ for~$\Pi$ such that:
\begin{enumerate} \itemsep=0pt
\item $F'(w)$ 
is a positive integer bounded by $\ell!$, for every
inference-vertex $w$ of $G(\Pi)$,
\item $G(\Pi,F)$ and~$G(\Pi,F')$ have the same sets of sources and
  sinks.
\end{enumerate}
\end{lemma}

\begin{proof}
Let $I$ and $J$ be the sets of inference- and formula-vertices of
$G(\Pi)$.  Let $S \subseteq J$ and $T \subseteq J$ be the sets of
sources and sinks of $G(\Pi,F)$, respectively. Consider the following
variant of the linear program $(P)$ above:
\begin{equation*}
(Q): \;\;
\begin{array}{lll}
\sum_{w \in N^-(u)} Y_w - \sum_{w \in N^+(u)} Y_w \geq 1 & & \text{ for each $u \in T$, } \\
\sum_{w \in N^-(u)} Y_w - \sum_{w \in N^+(u)} Y_w \geq 0 & \;\;\;\;\; & \text{ for each $u \in J\setminus (S \cup T)$, } \\
Y_w \geq 1 & & \text{ for each $w \in I$. }
\end{array}
\end{equation*}
When we transform~$(Q)$ it into standard form by adding
exactly~$|J|+|I|-|S|$ many slack variables, the result will be a
linear program of the form~$Mx = b$,~$x \geq 0$ where~$x$ is a vector
of~$2|I|+|J|-|S|$ variables,~$M$ is a constraint matrix of
dimensions~$(|I|+|J|-|S|) \times (2|I|+|J|-|S|)$, and~$b$ is a
right-hand side~$(|I|+|J|-|S|)$-vector. Moreover, each coefficient in the
matrix~$M$ and the vector~$b$ will be in~$\{-1,0,1\}$. Since this
linear program has a solution (the one given by~$F$ adequately
extended to the slack variables), it also has a basic feasible
solution~$(x^*_u)_{u \in V}$. Each component~$x^*_u$ is either~$0$ or,
by Cramer's Rule, can be written in the form~$\det(N_u)/\det(N)$
where~$N$ is a square submatrix of~$M$, and~$N_u$ is the matrix that
results from replacing the column of~$N$ of index~$u$ by a subvector
of the right-hand side vector~$b$. By ignoring the slack variables we
get a solution~$(y_w)_{w \in I}$ for~$(Q)$ of the same
form. Multiplying through by the common denominator~$\det(N)$ we get
an integral solution~$(y'_w)_{w \in I}$ for~$(Q)$ whose components
have the form~$\det(N_w)$; none is~$0$ because~$Y_w \geq 1$ is one of
the inequalities in~$(Q)$. Each~$N_w$-matrix has dimensions at
most~$(|I|+|J|-|S|) \times (|I|+|J|-|S|)$, and components in~$\{-1,0,1\}$. It
follows that~$y'_w = \det(N_w) \leq (|I|+|J|-|S|)!  =
\ell!$. Taking~$F'(w) := y'_w$ for each~$w \in I$ completes the proof.
\end{proof}

\subsection{Soundness of Circular Proofs}

In this section we develop the soundness proof when the set of
inference rules is fixed to axiom, symmetric cut, and split. See
Section~\ref{sec:prelim} for a discussion on this choice of rules. In
the next section we discuss the general case.

We give two different proofs: one combinatorial and one
(semi-)algebraic.

\begin{theorem} \label{thm:soundness} Let~$\mathscr{R}$ be the set of
  inference rules made of axiom, symmetric cut, and split.
  Let~$\mathscr{H}$ be a set of hypothesis formulas and let~$A$ be a
  goal formula. If there is a circular proof of~$A$ from~$\mathscr{H}$
  through the rules in~$\mathscr{R}$, then every truth assignment that
  satisfies every formula in~$\mathscr{H}$ also satisfies~$A$.
\end{theorem}

\begin{proof}[First proof]
  Fix a truth assignment~$\alpha$. We prove the stronger claim that,
  for every circular pre-proof~$\Pi$ from an unspecified set of
  hypothesis formulas, every integral flow assignment~$F$ for~$\Pi$,
  and every sink~$s$ of~$G(\Pi,F)$, if~$\alpha$ falsifies the formula
  that labels~$s$, then~$\alpha$ also falsifies the formula that
  labels some source of~$G(\Pi,F)$.
  The proof is by induction on the the sum of the flows assigned
  by~$F$, which we call the total flow-sum of~$F$.
  Such induction is possible because we restrict to integral flow
  assignments, which is without loss of generality by
  Lemma~\ref{lem:integralflow}.
 
  If the total flow-sum is zero, then there are no inferences, hence
  there are no sinks, and the statement holds vacuously. Assume then
  that the total flow-sum is positive, and let~$s$ be a sink
  of~$G(\Pi,F)$, with balance~$B(s) > 0$, whose labelling formula~$B$
  is falsified by~$\alpha$. Since its balance is positive,~$s$ must
  have at least one in-neighbour~$r$. Since the conclusion formula of
  the rule at~$r$ is falsified by~$\alpha$, some premise formula of
  the rule at~$r$ must exist that is also falsified
  by~$\alpha$. Let~$u$ be the corresponding in-neighbour of~$r$, and
  let~$B(u)$ be its balance.  If~$B(u)$ is negative, then~$u$ is a
  source of~$G(\Pi,F)$, and we are done. Assume then that~$B(u)$ is
  non-negative.

  Let~$\delta := \min\{B(s),F(r)\}$ and note that~$\delta > 0$
  because~$B(s) > 0$ and~$F(r) > 0$. We define a new circular
  pre-proof~$\Pi'$ and an integral flow assignment~$F'$ for~$\Pi'$ to
  which we will apply the induction hypothesis. The construction will
  guarantee the following properties:
  \begin{enumerate}
    \itemsep=0pt
    \item the total flow-sum of $F'$ is smaller than the total
    flow-sum of $F$.
    \item $u$ is a sink of $G(\Pi',F')$ and $s$ is not a source of
    $G(\Pi',F')$,
    \item if $t$ is a source of $G(\Pi',F')$, then $t$ is a source of
    $G(\Pi,F)$ or an out-neighbour of~$r$ in~$G(\Pi)$.
  \end{enumerate}
  {From} this the claim will follow by applying the induction
  hypothesis to~$\Pi'$,~$F'$ and~$u$. Indeed the induction hypothesis
  applies to them by Property~1 and the first half of Property~2. It
  will give a source~$t$ of~$G(\Pi',F')$ whose labelling formula is
  falsified by~$\alpha$.  We argue that~$t$ must also be a source
  of~$G(\Pi,F)$, in which case we are done. To argue for this, assume
  otherwise and apply Property~3 to conclude that~$t$ is an
  out-neighbour of~$r$ in~$G(\Pi)$, which by the second half of
  Property~2 must be different from~$s$ because~$t$ is a source
  of~$G(\Pi',F')$. Recall now that~$s$ is a second out-neighbour
  of~$r$. This can be the case only if~$r$ is a split inference, in
  which case the formulas that label~$s$ and~$t$ must be of the
  form~$C \vee B$ and~$C \vee \overline{B}$, respectively, for
  appropriate formulas~$C$ and~$B$.  But, by assumption,~$\alpha$
  falsifies the formula that labels~$s$, namely~$C \vee B$, which
  means that~$\alpha$ satisfies the formula~$C \vee \overline{B}$ that
  labels~$t$. This is the contradiction we were after.

  It remains to construct~$\Pi'$ and~$F'$ that satisfy
  Properties~1,~2, and~3. We define them by cases according to
  whether~$F(r) > B(s)$ or~$F(r) \leq B(s)$, and then argue for the
  correctness of the construction. In case~$F(r) > B(s)$, and
  hence~$\delta = B(s)$, let~$\Pi'$ be defined as~$\Pi$ without
  change, and let~$F'$ be defined by~$F'(r) := F(r)-\delta$
  and~$F'(w) := F(w)$ for every other~$w \in I \setminus
  \{r\}$. Obviously~$\Pi'$ is still a pre-proof and~$F'$ is an
  integral flow assignment for~$\Pi'$ by the assumption
  that~$F(r) > B(s) = \delta$.
  In case~$F(r) \leq B(s)$, and hence~$\delta = F(r)$, let~$\Pi'$ be
  defined as~$\Pi$ with the inference-step that labels~$r$ removed,
  and let~$F'$ be defined by~$F'(w) := F(w)$ for
  every~$w \in I \setminus \{r\}$. Note that in this case~$\Pi'$ is
  still a pre-proof but perhaps from a larger set of hypothesis
  formulas.

  In both cases the proof of the claim that~$\Pi'$ and~$F'$ satisfy
  Properties~1,~2, and~3 is the same. Property~1 follows from the fact
  that the total flow-sum of~$F'$ is the total flow-sum of~$F$
  minus~$\delta$, and~$\delta > 0$. The first half of Property~2
  follows from the fact that the balance of~$u$ in~$G(\Pi',F')$
  is~$B(u)+\delta$, while~$B(u) \geq 0$ by assumption
  and~$\delta > 0$. The second half of Property~2 follows from the
  fact that the balance of~$s$ in~$G(\Pi',F')$ is~$B(s)-\delta$,
  while~$B(s) \geq \delta$ by choice of~$\delta$. Property~3 follows
  from the fact that the only formula-vertices of~$G(\Pi',F')$ of
  balance smaller than that in~$G(\Pi,F)$ are the out-neighbours
  of~$r$. This completes the proof of the claim, and of the theorem.
\end{proof}

We give a second, different proof of soundness that will
play an important role later.

\begin{proof}[Second proof]
  Let~$\Pi$ be a circular pre-proof and let~$F$ be a flow assignment
  for~$\Pi$ that witnesses a proof of~$A$
  from~$\mathscr{H}$. Let~$\alpha$ be a truth assignment that
  satisfies all the formulas in~$\mathscr{H}$, and let~$s$ be an
  arbitrary formula-vertex in~$G(\Pi)$.  We show that if~$\alpha$
  falsifies the formula that labels~$s$, then~$s$ is not a sink
  of~$G(\Pi,F)$.

  Let~$V = I \cup J$ be the set of vertices of~$G(\Pi)$ partitioned
  into the set~$I$ of inference-vertices, and the set~$J$ of
  formula-vertices.  For every~$u \in J$, let~$A_u$ be the formula
  that labels~$u$ and let~$Z_u := \alpha(A_u)$; the truth-value
  that~$\alpha$ gives to~$A_u$. By inspection of the three allowed
  inference rules, for each~$w \in I$ with labelling inference
  rule~$R$ and in- and out-neighbours~$N^-$ and~$N^+$, respectively,
  we have:
  \begin{center} 
  \begin{tabular}{lcll} 
      $-(1-Z_a)$ & $\geq$ & $0$ & if $R = \text{axiom}$ with $N^+ =
      \{a\}$, \\
  $(1-Z_a) + (1-Z_b) - (1-Z_c)$ &
      $\geq$ & $0$ & if $R = \text{cut}$ with $N^- = \{a,b\}$ and
      $N^+ = \{c\}$, \\ 
  $(1-Z_a) - (1-Z_b) - (1-Z_c)$ & 
      $\geq$ & $0$ & if $R = \text{split}$ with $N^- = \{a\}$ and $N^+ = \{b,c\}$.
  \end{tabular}
  \end{center}
  Multiplying each such inequality by the positive flow~$F(w)$ of~$w$
  and adding up over all~$w \in I$ we get \begin{equation} \sum_{w \in
      I} F(w)\left({\sum_{v \in N^-(w)} (1-Z_v) - \sum_{u \in N^+(w)}
        (1-Z_u)}\right) \geq 0 \label{eqn:first} \end{equation}
  Rearranging the sum by formula-vertices, instead of arranging it by
  inference-vertices, we get \begin{equation} \sum_{u \in J}
    (1-Z_u)\left({\sum_{w \in N^+(u)} F(w) - \sum_{w \in N^-(u)}
        F(w)}\right) \geq 0. \label{eqn:second} \end{equation} The
  expression enclosed in parenthesis in the left-hand side
  in~\eqref{eqn:second} equals~$-B(u)$, where~$B(u)$ is the balance
  of~$u$ in~$G(\Pi,F)$. Now,~$Z_u = 1$ whenever~$u$ is a
  source,~$Z_s = 0$ for~$s$ by assumption, and~$B(u)(1-Z_u) \geq 0$
  for every other formula-vertex~$u \in J$ by the definition of
  circular proof. Hence \begin{equation} -B(s) \geq
    0, \label{eqn:third} \end{equation} which shows that~$s$ has
  non-positive balance in~$G(\Pi,F)$ and is thus not a
  sink.  \end{proof}

In the second proof of soundness, one can think of the~$Z_u$ as
variables that are constrained by: 1) the inequalities that express
the local soundness of the three types of inference rules, 2) the
equations~$Z_u = 1$ for~$u$ a source, that express that each
hypothesis is satisfied, and 3) the equation~$Z_s = 0$ for~$s$ a sink,
that expresses that some conclusion is falsified. In other words, the
constraints of types 1), 2) and 3) express that the proof is unsound,
while equations~\eqref{eqn:first} and~\eqref{eqn:second} say that any
valid flow assignment~$F(w)$ can play the role of a witness of its
infeasibility.

\subsection{Soundness for Other Sets of Rules}

We claim that both proofs of soundness that we gave apply without
change to any set of sound inference rules that have a single
conclusion formula.  This requirement is fulfilled by all sets of
standard inference rules, such as~\eqref{eqn:Fregerules}, and is
subsumed by the following more general but technical one:
\begin{quotation}
  \noindent (*) Any inference rule in~$\mathscr{R}$ that has more than
  one conclusion formula has the property that any truth assignment
  that falsifies one of its conclusion formulas must satisfy all other
  conclusion formulas.
\end{quotation}
Obviously, if all rules in $\mathscr{R}$ have a single conclusion, then
(*) is satisfied. Note also that the only rule that has more than one
conclusion formula among axiom, symmetric cut, and split is split, and
clearly it has the required property. Thus, the following statement
generalizes Theorem~\ref{thm:soundness}.

\begin{theorem} \label{thm:soundnessgeneral} Let~$\mathscr{R}$ be a
  set of sound inference rules that satisfy property
  (*). Let~$\mathscr{H}$ be a set of hypothesis formulas and let~$A$
  be a goal formula. If there is a circular proof of~$A$
  from~$\mathscr{H}$ through the rules in~$\mathscr{R}$, then every
  truth assignment that satisfies every formula in~$\mathscr{H}$ also
  satisfies~$A$.
\end{theorem}

\begin{proof}
  The first proof of Theorem~\ref{thm:soundness} was already phrased
  in a way that the generalization to sets of inference rules that
  satisfy (*) is straightforward. We discuss the generalization of the
  second proof. Let~$\Pi$ be a circular proof with rules
  in~$\mathscr{R}$, let~$A_u$ be the formula that labels the
  formula-vertex~$u$, let~$Z_u := \alpha(A_u)$ be the truth value
  given to~$A_u$ by a truth assignment~$\alpha$, and let~$w$ be an
  inference-vertex of~$\Pi$ with in- and out-neighbors~$N^-$
  and~$N^+$, respectively.  Then, the following inequality holds:
\begin{equation}
\sum_{u \in N^-} (1-Z_u) - \sum_{u \in N^+} (1-Z_u) \geq 0
\label{eqn:generalizedequation}
\end{equation}
Indeed, if~$Z_a = 0$ for some~$a \in N^+$, then by the soundness of
the rule there exists~$b \in N^-$ such that~$Z_b = 0$, and by (*) we
also have~$Z_c = 1$ for every~$c \in N^+\setminus\{a\}$. The
conclusion to this is that the left summand
in~\eqref{eqn:generalizedequation} is at least~$1$, and the right
summand in~\eqref{eqn:generalizedequation} is exactly~$1$, so their
difference is non-negative.  {From} here it suffices to note that this
is the only property we used in order to derive
equations~\eqref{eqn:first},~\eqref{eqn:second} and~\eqref{eqn:third}.
\end{proof}

\commentout{

\subsection{Old}

Let $\Gamma$ be a set of propositional formulas. A
\introduceterm{circular $\Gamma$-proof} is a labeled bipartite
directed graph whose edges join \introduceterm{rule vertices}, labeled
by positive integers, and \introduceterm{formula vertices}, labeled by
formulas that are disjunctions of formulas in $\Gamma$, in such a way
that the closed neighborhood of every rule vertex is a copy of one of
the two allowed rules: the \introduceterm{cut rule} and the
\introduceterm{split rule} (see Figures~\ref{fig:cut}
and~\ref{fig:split}).

\begin{figure}[h]
  \begin{subfigure}[t]{0.5\textwidth} \centering \begin{tikzpicture} \path[draw=none,use
    as bounding box] (-3,-2.5) rectangle (3,2.5);
      \node[clause] (left)  at (-2,1.5) {$F \vee \bar{A}$};
      \node[clause] (right) at (+2,1.5) {$F \vee A$};
      \node[rule]   (rule)  at ( 0,0) {$n$};
      \node[clause] (conclusion) at (0,-1.5) {$F$};
      \draw[edge]   (rule)  to (conclusion);
      \draw[edge]   (left)  to (rule);
      \draw[edge]   (right) to (rule);
    \end{tikzpicture}
    \caption{Cut rule.}
    \label{fig:cut}
  \end{subfigure}
  \begin{subfigure}[t]{0.5\textwidth} \centering \begin{tikzpicture} \path[draw=none,use
    as bounding box] (-3,-2.5) rectangle (3,2.5);
      \node[clause] (left)  at (-2,-1.5) {$F \vee \bar{A}$};
      \node[clause] (right) at (+2,-1.5) {$F \vee A$};
      \node[rule]   (rule)  at ( 0,0) {$n$};
      \node[clause] (start) at (0,+1.5) {$F$};
      \draw[edge]   (start) to (rule);
      \draw[edge]   (rule)  to (left);
      \draw[edge]   (rule)  to (right);
    \end{tikzpicture}
    \caption{Split rule.}
    \label{fig:split}
  \end{subfigure}
  \caption{Rules for circular proofs.}
  \label{fig:circularproofs}
\end{figure}

The bipartite graph in a circular proof is called its
\introduceterm{proof-graph}, and here bipartite means that every in-
or out-neighbour of a formula vertex is a rule vertex, and
vice-versa. The positive integer labeling a rule vertex is called its
\introduceterm{flow}, and the formula labeling a formula vertex is
called its \introduceterm{labeling formula}.  The
\introduceterm{inflow} of a formula vertex is the sum of the flows of
its in-neighbours. Similarly, the \introduceterm{outflow} of a formula
vertex is the sum of the flows of its out-neighbours. The
\introduceterm{balance} of a formula vertex is the difference between
its inflow and its outflow. The formulas that label vertices of
strictly negative balance are called \introduceterm{sources} and those
that label vertices of strictly positive balance are called
\introduceterm{sinks}.

Let $\mathcal{H} \cup \{F\}$ be a set of formulas and let $\Pi$ be a circular
$\Gamma$-proof. We say that $\Pi$ is a \introduceterm{circular
  $\Gamma$-proof of $F$ from $\mathcal{H}$}  if all sources of $\Pi$ are either
logical axioms or formulas from $\mathcal{H}$, and the formula $F$ appears among
the sinks of $\Pi$. We say that $\Pi$ is a circular
$\Gamma$-\introduceterm{refutation} of $\mathcal{H}$ if it is a circular
$\Gamma$-proof of the empty (i.e.\ false) formula from $\mathcal{H}$. The
\introduceterm{length} of $\Pi$ is the number of rule vertices in the
proof-graph, and its \introduceterm{size} is the sum of the sizes of
the labeling formulas and the bit-lengths of the flows.

A \introduceterm{circular Frege proof} is a circular $\Gamma$-proof
where $\Gamma$ is the set of all propositional formulas. A
\introduceterm{circular resolution proof} is a circular $\Gamma$-proof
where $\Gamma$ is the set of all literals. See
Figure~\ref{fig:circularresolution}.

\begin{figure}[h]
  \begin{subfigure}[t]{0.5\textwidth} \centering \begin{tikzpicture} \path[draw=none,use
    as bounding box] (-3,-2.5) rectangle (3,2.5);
      \node[clause] (left)  at (-2,1.5) {$C \vee \bar{x}$};
      \node[clause] (right) at (+2,1.5) {$C \vee x$};
      \node[rule]   (rule)  at ( 0,0) {$n$};
      \node[clause] (conclusion) at (0,-1.5) {$C$};
      \draw[edge]   (rule)  to (conclusion);
      \draw[edge]   (left)  to (rule);
      \draw[edge]   (right) to (rule);
    \end{tikzpicture}
    \caption{Resolution rule.}
    \label{fig:rescut}
  \end{subfigure}
  \begin{subfigure}[t]{0.5\textwidth} \centering \begin{tikzpicture} \path[draw=none,use
    as bounding box] (-3,-2.5) rectangle (3,2.5);
      \node[clause] (left)  at (-2,-1.5) {$C \vee \bar{x}$};
      \node[clause] (right) at (+2,-1.5) {$C \vee x$};
      \node[rule]   (rule)  at ( 0,0) {$n$};
      \node[clause] (start) at (0,+1.5) {$C$};
      \draw[edge]   (start) to (rule);
      \draw[edge]   (rule)  to (left);
      \draw[edge]   (rule)  to (right);
    \end{tikzpicture}
    \caption{Split rule.}
    \label{fig:ressplit}
  \end{subfigure}
  \caption{Rules for circular resolution.}
  \label{fig:circularresolution}
\end{figure}

The definition of circular proofs allows for some simplifications and
further normalizations of the proof-graph without loss of
generality. In particular we can always assume that there are no two
formula vertices labeled with the same formula\commentAA{I'm not
  terribly fond of this because then tree-like proofs cannot be just
  put as the special case in which the proof-graph is a tree; well,
  sure they can, but then it is not true that ``we can always
  assume...''.}, therefore we can abuse terminology and talk about the
balance of a formula\commentAA{Accordingly, I tried to avoid this
  terminology by referring to nodes of the proof graph and explictly
  referring to its labeling formula instead of referring to formulas
  directly; I'm not sure I always managed, but I think it is doable.}.
In general we could also assume that the graph is connected, and that
it contains at least one conclusion. Nevertheless in the soundness
proof it may be convenient to allow degenerate proofs where the there
are multiple connected components, some of which without any
conclusion (e.g., an isolated formula vertex).\commentML{I like
  degenerate proofs to make sense anyway.}

}

\section{Circular Resolution} \label{sec:resolution} 

In this section we investigate the power of Circular
Resolution. Recall from the discussion in
Section~\ref{sec:symmetricrules} that Resolution is traditionally
defined to have cut as its only rule, but that an essentially
equivalent version of it is obtained if we define it through symmetric
cut, split, and axiom, still all restricted to clauses. This more
liberal definition of Resolution, while staying equivalent vis-a-vis
the tree-like and dag-like versions of Resolution, will play an
important role for the circular version of Resolution.

While for Frege proof systems we will prove later that there is no qualitative
difference between tree-like, dag-like, and circular proofs, in this
section we show that circular Resolution can be exponentially stronger
than dag-like Resolution. Indeed, we show that Circular Resolution is
polynomially equivalent with the Sherali-Adams proof system, which is
already known to be stronger than dag-like Resolution:

\begin{theorem} \label{thm:circularresolutionvsSA} Sherali-Adams and
  Circular Resolution polynomially simulate each other.  Moreover, the
  simulation one way converts monomial size~$s$ and degree~$d$ into
  size~$O(s)$ and width~$d$, and the simulation in the reverse way
  converts size~$s$ and width~$w$ into monomial size~$O(s)$ and
  degree~$w$.
\end{theorem}

For the statement of Theorem~\ref{thm:circularresolutionvsSA} to even
make sense, Sherali-Adams is to be understood as a proof system for
deriving clauses from clauses, under an appropriate encoding of
clauses as polynomial inequalities, to be discussed later in this
section.

\commentout{

We discuss the restriction of circular Frege which only deals with
clauses, which is essentially a circular version of resolution.

A \introduceterm{resolution derivation} of a clause $C$ from a CNF
formula $F=\bigwedge_{j=1}^{m} C_{j}$ over variables
$x_{1}, \ldots, x_{n}$ is a sequence of clauses
$(D{1}, \ldots, D_{s})$ such that $D_{s}=C$ and such that each $D_{t}$
in the sequence is either
\begin{itemize}
  \item (axioms) a clauses $C_{j}$ for $j\in[m]$; or
  \item (logical axioms) a clause $ x_{i} \lor \overline{x_{i}} $ for some $i \in [n]$; or
  \item (resolution) $D_{t} = A \lor B$ where there are
  $D_{t'} = A \lor x_{i} $ and $D_{t''}= B \lor \overline{x_{i}}$
  for some $0 < t',t'' < t$ and $i \in [n]$;
  \item (weakening) $D_{t} = D_{t'} \lor B$ for some $0 < t' < t$.
\end{itemize}
The size of such proof is $s$ and the width of such proof is
$\max_{t=1}^{s}|D_{t}|$. A \introduceterm{resolution refutation} of
$F$ is a resolution proof of the empty clause from $F$.\footnote{If we
  only care about resolution refutation then we do not need logical
  axioms or weakening, but these are necessary to guarantee
  completeness over clauses.}

Recall that circular resolution proof from a CNF formula
$F=\bigwedge_{j=1}^{m} C_{j}$ is a circular Frege proof of $C$ from
$\{C_{j}\}_{j\in [m]} \cup \{x_{i} \lor \overline{x_{i}}\}_{i \in
  [n]}$ where all formula vertex are labeled by clauses.
Notice that circular resolution does not have a weakening rule as
resolution, but this is not an issue since the split rule can be used
to simulate it efficiently. 
\begin{claim}[Weakening]
  \label{clm:weakening}
  For every $n>0$ there exists a circular resolution proof of size
  $|A \lor B|$ where
  \begin{itemize}
    \item clauses $A$ has balance $-n$;
    \item clauses $A \vee B$ has balance $n$;
    \item all other clauses in the proof have non negative balance.
  \end{itemize}
\end{claim}
\begin{proof}
  The proof is just a sequence of split rules. Without loss of
  generality assume that $B/A=b_{1} \lor b_{2} \lor \cdots \lor b_{k}$
  and that $A$ and $B$ do not contain opposite literals.
  From $A$ we apply the split rule to get $A \lor b_{1}$ and
  $A \lor \overline{b_{1}}$, we ignore the latter and from
  $A \lor b_{1}$ we apply the split rule to get
  $A \lor b_{1} \lor b_{2}$ and $A \lor b_{1} \lor \overline{b_{2}}$.
  We proceed in the same way until we get to $A \lor b_{1} \lor b_{2}
  \lor \cdots \lor b_{k}$ which is $A \lor B$.
  For all the split rules the flow is set to be equal to $n$, so that
  the claims about the balance are satisfied.
\end{proof}

The following corollary then follows immediately.
\begin{corollary} 
  For any resolution proof of $C$ from $F$ of size $s$ and width $w$
  there exists a circular resolution proof of $C$ from $F$ of size
  $O(w \cdot s)$ and width $w+1$.
\end{corollary}
\begin{proof}[Proof sketch]
  All applications of the resolution rule that get clause $A \lor B$
  form $A \lor x$ and $B \lor \overline{x}$ can be turned into the
  forms that get clauses $A \lor B$ form $A \lor B \lor x$ and
  $A \lor B \lor \overline{x}$ using weakening and increasing the
  width of the proof to $w+1$. Now we turn this resolution proof into
  a circular resolution proof just by converting every resolution step
  into the corresponding rule application from Figure~\ref{fig:rescut}.
  What remains it to simulate the weakening rules, which can be done
  using Claim~\ref{clm:weakening}.
  The only missing step is set the flow on the resolution rules. To do
  that start from the bottom setting the flow of the last rule vertex
  to $1$ and there, since the proof graph is a DAG, set the numbers in
  reverse topological order so that the balance constraints on non
  axiom clauses is non negative. 
\end{proof}

}

\subsection{Pigeonhole Principles} \label{sec:exponentialseparation}

\newcommand{\PHP}{\mathrm{PHP}}

We start by showing that the Pigeonhole Principle
formula~$\PHP^{n+1}_n$ has small Circular Resolution proofs. By the
well-known lower bound of Haken \cite{Haken1985}, this will show that
Circular Resolution is exponentially stronger than Resolution.
However, we show the stronger claim that Circular Resolution is also
stronger than Resolution when measured in terms of \emph{width}; i.e.,
the length of the longest clause in the proof.  To prove this, we need
to introduce the \emph{bipartite graph-based} variant of the
Pigeonhole Principle from~\cite{BenSassonWigderson2001}.
 
Let~$G$ be a bipartite graph with vertex bipartition~$(U,V)$, and set
of edges~$E \subseteq U \times V$. For a vertex~$w \in U \cup V$, we
write~$N_G(w)$ to denote the set of neighbours of~$w$ in~$G$,
and~$\deg_G(w)$ to denote its degree. The Graph Pigeonhole Principle
of~$G$, denoted by~$G$-$\PHP$, is a CNF formula that has one
variable~$X_{u,v}$ for each edge~$(u,v)$ in~$E$ and the following set
clauses:
\smallskip

\begin{center}
\begin{tabular}{ll}
$X_{u,v_1} \vee \cdots \vee X_{u,v_d}$ & for $u \in U$ with
  $N_G(u) = \{v_1,\ldots,v_d\}$, \\ 
$\overline{X_{u_1,v}} \vee
  \overline{X_{u_2,v}}$ & for $u_1,u_2 \in U$, $v \in V$
  with $u_1 \not= u_2$, and $v \in N_G(u_1) \cap N_G(u_2)$.
\end{tabular}
\end{center}
\smallskip

If~$|U| > |V|$, and in particular if~$|U|=n+1$ and~$|V|=n$,
then~$G$-$\PHP$ is unsatisfiable by the pigeonhole principle.
For~$G = K_{n+1,n}$, the complete bipartite graph with sides of
sizes~$n+1$ and~$n$, the formula~$G$-$\PHP$ is the standard CNF
encoding~$\PHP^{n+1}_n$ of the pigeonhole principle.

Even for certain constant degree bipartite graphs with~$|U|=n+1$
and~$|V|=n$, the formulas are hard for Resolution.

\begin{theorem}[\cite{BenSassonWigderson2001,Haken1985}] \label{thm:bswandhaken}
  There are families of bipartite graphs~$(G_{n})_{n \geq 1}$,
  where~$G_n$ has maximum degree bounded by a constant and vertex
  bipartition~$(U,V)$ of~$G_n$ that satisfies~$|U| = n+1$
  and~$|V| = n$, such that every Resolution refutation of~$G_{n}$-$\PHP$
  has width~$\Omega(n)$ and length~$2^{\Omega(n)}$. Moreover, this
  implies that every Resolution refutation of~$\PHP^{n+1}_n$ has
  length~$2^{\Omega(n)}$.
\end{theorem}

In contrast, we show that these formulas have Circular Resolution
refutations of polynomial length and, simultaneously, constant width.

\begin{theorem} \label{thm:upperbound} For every bipartite graph~$G$
  of maximum degree~$d$ with bipartition~$(U,V)$ such that~$|U|>|V|$,
  there is a Circular Resolution refutation of~$G$-$\PHP$ of length
  polynomial in~$|U|+|V|$ and width~$d$. 
\end{theorem}
\begin{proof}
  We build the graph of the refutation in parts.
  Concretely, for every~$u \in U$ and~$v \in V$, we describe two
  Circular Resolution proofs~$\Pi_{u\rightarrow}$
  and~$\Pi_{\rightarrow v}$, with their associated flow assignments.
  These proofs will have width bounded by~$\deg_G(u)$ and~$\deg_G(v)$,
  respectively, and size polynomial in~$\deg_G(u)$ and~$\deg_G(v)$,
  respectively. Moreover, the following properties will be ensured:
\begin{enumerate} \itemsep=0pt
\item The proof-graph of~$\Pi_{u\rightarrow}$ contains a
  formula-vertex labelled by the empty clause~$\emptyclause$ with
  balance~$+1$ and a formula-vertex labelled~$\overline{X_{u,v}}$ with
  balance~$-1$ for every~$v \in N_G(u)$; any other formula-vertex that
  has negative balance is labelled by a clause of~$G$-$\PHP$.
\item The proof-graph~$\Pi_{\rightarrow v}$ contains a formula-vertex
  labelled by the empty clause~$\emptyclause$ with balance~$-1$ and a
  formula-vertex labelled by~$\overline{X_{u,v}}$ with balance~$+1$
  for every~$u \in N_G(v)$; any other formula-vertex that has negative
  balance is labelled by a clause of~$G$-$\PHP$.
\end{enumerate}
By merging~$n$ of the~$n+1$ many formula-vertices labelled by the
empty clause in proofs of the form~$\Pi_{u \rightarrow}$ with the~$n$
many formula-vertices labelled by the empty clause in proofs of the
form~$\Pi_{\rightarrow v}$ we get a proof in which all the
formula-vertices that have negative balance are clauses of~$G$-$\PHP$,
and the empty clause~$\emptyclause$ has positive balance. This is
indeed a Circular Resolution refutation of~$G$-$\PHP$. See
Figure~\ref{fig:proofofphp43} for a diagram of the proof
for~$\PHP^4_3$.

\begin{figure}
  \begin{center}
  \begin{tikzpicture}
  \draw[help lines, white] (0,0) grid (11,9); 

  \node at (5.5,8) {$\overline{X_{11}}$};
  \node at (5.5,7.5) {$\overline{X_{12}}$};
  \node at (5.5,7) {$\overline{X_{13}}$};

  \node at (5.5,6) {$\overline{X_{21}}$};
  \node at (5.5,5.5) {$\overline{X_{22}}$};
  \node at (5.5,5) {$\overline{X_{23}}$};

  \node at (5.5,4) {$\overline{X_{31}}$};
  \node at (5.5,3.5) {$\overline{X_{32}}$};
  \node at (5.5,3) {$\overline{X_{33}}$};

  \node at (5.5,2) {$\overline{X_{41}}$};
  \node at (5.5,1.5) {$\overline{X_{42}}$};
  \node at (5.5,1) {$\overline{X_{43}}$};

  \node at (9.2,8) {$0$};
  \node at (9.2,6) {$0$};
  \node at (9.2,4) {$0$};
  \node at (9.2,2) {$0$};

  \draw[line width = 1pt, ->] plot [smooth] coordinates { (9.5,8) (9.6,8.4) (1.9,8.4) (2,8) };
  \draw[line width = 1pt, ->] plot [smooth] coordinates { (9.5,6) (9.6,6.4) (1.9,6.4) (2,6) };
  \draw[line width = 1pt, ->] plot [smooth] coordinates { (9.5,4) (9.6,4.4) (1.9,4.4) (2,4) };

  \draw[line width = 1pt, black,->] (6.2, 8.0) -- (8.0, 8.0);
  \draw[line width = 1pt, black,->] (6.2, 7.5) -- (8.0, 7.9);
  \draw[line width = 1pt, black,->] (6.2, 7.0) -- (8.0, 7.85);
  \draw[line width = 1pt, black] (8.25,8) circle (5pt);
  \draw[line width = 1pt, black] (8.25,8) circle (3pt);
  \draw[line width = 1pt, black, ->] (8.5,8) -- (8.8,8);

  \draw[line width = 1pt, black,->] (6.2, 6.0) -- (8.0, 6.0);
  \draw[line width = 1pt, black,->] (6.2, 5.5) -- (8.0, 5.9);
  \draw[line width = 1pt, black,->] (6.2, 5.0) -- (8.0, 5.85);
  \draw[line width = 1pt, black] (8.25, 6) circle (5pt);
  \draw[line width = 1pt, black] (8.25, 6) circle (3pt);
  \draw[line width = 1pt, black, ->] (8.5, 6) -- (8.8, 6);

  \draw[line width = 1pt, black,->] (6.2, 4.0) -- (8.0, 4.0);
  \draw[line width = 1pt, black,->] (6.2, 3.5) -- (8.0, 3.9);
  \draw[line width = 1pt, black,->] (6.2, 3.0) -- (8.0, 3.85);
  \draw[line width = 1pt, black] (8.25, 4) circle (5pt);
  \draw[line width = 1pt, black] (8.25, 4) circle (3pt);
  \draw[line width = 1pt, black, ->] (8.5, 4) -- (8.8, 4);

  \draw[line width = 1pt, black,->] (6.2, 2.0) -- (8.0, 2.0);
  \draw[line width = 1pt, black,->] (6.2, 1.5) -- (8.0, 1.9);
  \draw[line width = 1pt, black,->] (6.2, 1.0) -- (8.0, 1.85);
  \draw[line width = 1pt, black] (8.25, 2) circle (5pt);
  \draw[line width = 1pt, black] (8.25, 2) circle (3pt);
  \draw[line width = 1pt, black, ->] (8.5, 2) -- (8.8, 2);

  \draw[line width = 1pt, black] (2.3, 8) circle (5pt);
  \draw[line width = 1pt, black] (2.3, 8) circle (3pt);
  \draw[line width = 1pt, black, ->] (2.75, 8) -- (5, 8);
  \draw[line width = 1pt, black, ->] (2.75, 7.9) -- (5, 6);
  \draw[line width = 1pt, black, ->] (2.75, 7.8) -- (5, 4);
  \draw[line width = 1pt, black, ->] (2.75, 7.7) -- (5, 2);

  \draw[line width = 1pt, black] (2.3, 6) circle (5pt);
  \draw[line width = 1pt, black] (2.3, 6) circle (3pt);
  \draw[line width = 1pt, black, ->] (2.75, 6.1) -- (5, 7.5);
  \draw[line width = 1pt, black, ->] (2.75, 6.0) -- (5, 5.5);
  \draw[line width = 1pt, black, ->] (2.75, 5.9) -- (5, 3.5);
  \draw[line width = 1pt, black, ->] (2.75, 5.8) -- (5, 1.5);

  \draw[line width = 1pt, black] (2.3, 4) circle (5pt);
  \draw[line width = 1pt, black] (2.3, 4) circle (3pt);
  \draw[line width = 1pt, black, ->] (2.75, 4.2) -- (5, 7);
  \draw[line width = 1pt, black, ->] (2.75, 4.1) -- (5, 5);
  \draw[line width = 1pt, black, ->] (2.75, 4.0) -- (5, 3);
  \draw[line width = 1pt, black, ->] (2.75, 3.9) -- (5, 1);

  \end{tikzpicture} \end{center} \caption{The diagram of the circular
  proof of $\PHP^4_3$. The double circles indicate multiple
  inferences. The empty clause $0$ is derived four times and used only
  three times.} \label{fig:proofofphp43}
\end{figure}

For the construction of~$\Pi_{u\rightarrow}$, rename
the neighbours of~$u$ as~$1,2,\ldots,\ell$.
Let~$C_j$ denote the clause~$X_{u,1} \vee \cdots \vee X_{u,j}$ and
note that~$C_\ell$ is a clause of~$G$-$\PHP$.
Split~$\overline{X_{u,\ell}}$ on~$X_{u,1}$, then
split~$\overline{X_{u,\ell}} \vee X_{u,1}$ on~$X_{u,2}$, then
split~$\overline{X_{u,\ell}} \vee X_{u,1} \vee X_{u,2}$ on~$X_{u,3}$,
and so on until we produce~$\overline{X_{u,\ell}} \vee C_{\ell-1}$.
Then resolve this clause with~$C_\ell$ to produce~$C_{\ell-1}$. As a
sequence, this part of~$\Pi_{u\rightarrow}$ looks as follows:
\begin{subequations}
\label{eqn:phpdeduced}
\begin{align}
 & \overline{X_{u,\ell}}  \\
 & \overline{X_{u,\ell}} \vee X_{u,1} \\
 & \overline{X_{u,\ell}} \vee X_{u,1} \vee X_{u,2} \\
 &  \vdots\\
 & \overline{X_{u,\ell}} \vee X_{u,1} \vee X_{u,2} \vee \ldots \vee X_{u,{\ell-1}} \\
 & X_{u,1} \vee X_{u,2} \vee \ldots \vee X_{u,\ell-1}.
\end{align}
\end{subequations}
We now repeat essentially the same construction: we start
with~$\overline{X_{u,\ell-1}}$ and we split on~$X_{u,1}$,~$X_{u,2}$,
\ldots, as before until we
produce~$\overline{X_{u,\ell-1}} \vee C_{\ell-2}$. Cutting the latter
with the previously deduced~$C_{\ell-1}$ gives~$C_{\ell-2}$.
We continue other~$\ell-2$ times in order to get down to the empty
clause.
Observe that we can set the flow of all splits and cuts in this proof
to~$+1$ to get the balance claimed above.

For the construction of~$\Pi_{\rightarrow v}$ we need some more work.
Again rename the neighbours of~$v$ as~$1,2,\ldots,\ell$. All the
inference steps in this construction will have assigned flow~$+1$.
The first step in building the proof~$\Pi_{\rightarrow v}$ is the
derivation of the following sequence of clauses:
\begin{subequations}
\label{eqn:phpdeduced}
\begin{align}
 & \overline{X_{1,v}} \label{eqn:phpdeducedF} \\
 & X_{1,v} \vee \overline{X_{2,v}}\\
 & X_{1,v} \vee X_{2,v} \vee \overline{X_{3,v}}\\
 &  \vdots\\
 & X_{1,v} \vee X_{2,v}\vee \ldots \vee \overline{X_{\ell-1,v}}\\
 & X_{1,v} \vee X_{2,v}\vee \ldots \vee X_{\ell-1,v}. \label{eqn:phpdeducedL}
\end{align}
\end{subequations}
Notice that~\eqref{eqn:phpdeducedL} does not follow the pattern of the
previous clauses.
We start these inferences by splitting the empty clause on
variable~$X_{1,v}$, to get clauses~$\overline{X_{1,v}}$
and~$X_{1,v}$. We split the latter on~$X_{2,v}$ to
get~$X_{1,v} \vee \overline{X_{2,v}}$ and~$X_{1,v} \vee
X_{2,v}$. For~$i=3,\ldots,\ell-1$ we keep
splitting~$X_{1,v} \vee \ldots \vee X_{i-1,v}$ on~$X_{i,v}$ to get the
clauses~$X_{1,v} \vee \ldots \vee X_{i-1,v} \vee \overline{X_{i,v}}$
and~$X_{1,v} \vee \ldots \vee X_{i-1,v} \vee X_{i,v}$.
The empty clause has balance~$-1$ and all clauses
(\ref{eqn:phpdeducedF}--\ref{eqn:phpdeducedL}) have balance~$+1$.

Now to complete the proof~$\Pi_{\rightarrow v}$ we deduce the~$\ell$
singleton clauses~$\overline{X_{1,v}}, \ldots, \overline{X_{\ell,v}}$
from the clauses (\ref{eqn:phpdeducedF}--\ref{eqn:phpdeducedL}).
Observe that a simple sequence of cut rules between a
clause~$X_{1,v} \vee \ldots \vee X_{i-1,v} \vee \overline{X_{i,v}}$
and clauses~$\overline{X_{j,v}} \vee \overline{X_{i,v}}$
for~$1 \leq j < i$ produces the singleton clause~$\overline{X_{i,v}}$,
furthermore its balance is~$+1$.
We do a similar sequence of steps
using~$X_{1,v} \vee \ldots \vee X_{\ell-2,v} \vee X_{\ell-1,v}$ and
clauses~$\overline{X_{j,v}} \vee \overline{X_{\ell,v}}$
for~$1 \leq j < \ell$ to produce the clause~$\overline{X_{\ell,v}}$
with balance~$+1$, and conclude the construction
of~$\Pi_{\rightarrow v}$.
\end{proof}

\commentout{
\begin{figure}
\begin{center}
\begin{tikzpicture} \draw[help lines, gray, use as bounding box] (0,0) grid
 (12,3);
\node[formulavertexnobox] (a1) at (0.8,1.5) {$\scriptstyle\overline{X_{u,\ell}}$};
\node[formulavertexnobox] (a2) at (3,1.5) {$\scriptstyle \overline{X_{u,\ell}} \vee X_{u,1}$};
\node[formulavertexnobox] (adots) at (5.2,1.5) {$\scriptstyle \cdots$};
\node[formulavertexnobox] (al) at (6.1,1.5) {$\scriptstyle \overline{X_{u,\ell}} \vee C_{\ell-1}$};

\node[formulavertexnobox] (a1p) at (1.8,0.5) {$\scriptstyle\overline{X_{u,\ell-1}}$};
\node[formulavertexnobox] (a2p) at (4.35,0.5) {$\scriptstyle \overline{X_{u,\ell-1}} \vee X_{u,1}$};
\node[formulavertexnobox] (adotsp) at (6.7,0.5) {$\scriptstyle \cdots$};
\node[formulavertexnobox] (alp) at (7.81,0.5) {$\scriptstyle \overline{X_{u,\ell-1}} \vee C_{\ell-2}$};

\node[formulavertexnobox] (cl) at (6,2.5) {$\scriptstyle C_\ell$};
\node[formulavertexnobox] (clm1) at (8,2.5) {$\scriptstyle C_{\ell-1}$};
\node[formulavertexnobox] (clm2) at (10,2.5) {$\scriptstyle C_{\ell-1}$};
\node[formulavertexnobox] (cdots) at (11.1,2.5) {$\scriptstyle\cdots$};
\node[formulavertexnobox] (empty) at (11.5,2.5) {$\scriptstyle\emptyclause$};

\node[inferencevertexsmall] (r1) at (1.65,1.5) {};
\node[inferencevertexsmall] (r1i) at (4.35,1.5) {};
\node[inferencevertexsmall] (r1p) at (2.85,0.5) {};
\node[inferencevertexsmall] (r1ip) at (5.85,0.5) {};
\node[inferencevertexsmall] (r2) at (7,2.5) {};
\node[inferencevertexsmall] (r3) at (9,2.5) {};

\draw[edge] (a1) to (r1);
\draw[edge] (r1) to (a2);
\draw[edge] (al) to (r2);
\draw[edge] (cl) to (r2);
\draw[edge] (r2) to (clm1);
\draw[edge] (r1i) to (adots);
\draw[edge] (r1ip) to (adotsp);
\draw[edge] (clm1) to (r3);
\draw[edge] (alp) to (r3);
\draw[edge] (r3) to (clm2);
\draw[edge] (clm2) to (cdots);
\draw[edge] (a2) to (r1i);
\draw[edge] (a2p) to (r1ip);
\draw[edge] (a1p) to (r1p);
\draw[edge] (r1p) to (a2p);
\end{tikzpicture}
\end{center}
\caption{The proof $\Pi_{u \to}$ for $u \in U$. 
The conclusions of the split rules that are not needed are not drawn.
Every inference-vertex is assigned flow $1$.}
\label{fig:pigeonproof}
\end{figure}

\begin{figure}
\begin{center}
\begin{tikzpicture} \draw[help lines, gray, use as bounding box] (0,0) grid
 (16,3);
\node[formulavertex] (a1) at (1,1) {$A_1$};
\node[inferencevertex] (r1) at (1.5,1) {$R_1$};
\draw[edge] (a1) to (r1);
\end{tikzpicture}
\end{center}
\caption{The proof $\Pi_{u \to}$ for $u \in U$. All flows are $1$.}
\label{fig:holeproof}
\end{figure}
}

We get the following immediate consequence:

\begin{corollary} \label{cor:upperbound}
  The pigeonhole formulas~$\PHP^{n+1}_n$ have Circular Resolution
  refutations of polynomial length.
\end{corollary}

In combination with Theorem~\ref{thm:bswandhaken}, this shows that
Circular Resolution can be exponentially stronger than Resolution.  In
the next two sections we determine the exact power of Circular
Resolution.

\subsection{Simulation by Sherali-Adams} \label{sec:simulationby}

In this section we prove one half of
Theorem~\ref{thm:circularresolutionvsSA}. We need some preparation.
Fix a set of variables~$X_1,\ldots,X_n$ and their
twins~$\bar{X}_1,\ldots,\bar{X}_n$.  For a
clause~$C = \bigvee_{j \in Y} X_j \vee \bigvee_{j \in Z}
\overline{X_j}$,
define
\begin{equation}
T(C) :=  - \prod_{j \in Y} \bar{X}_j \prod_{j \in Z} X_j,
\label{eqn:multiplicative}
\end{equation}
Observe that a truth assignment satisfies~$C$ if and only if the
corresponding 0-1 assignment for the variables of~$T(C)$ makes the
inequality~$T(C) \geq 0$ true. There is an alternative encoding of
clauses into inequalities that is sometimes used. Define
\begin{equation}
L(C) := \sum_{j \in Y} X_j + \sum_{j \in Z} \bar{X}_j - 1,
\label{eqn:additive}
\end{equation}
and observe that a truth assignment satisfies~$C$ if and only if the
corresponding 0-1 assignment makes the inequality~$L(C) \geq 0$ true.
We state the results of this section for the~$T$-encoding of clauses,
but the same result would hold for the~$L$-encoding because there is
an efficient SA proof of~\eqref{eqn:multiplicative}
from~\eqref{eqn:additive} (see Lemma~4.2 in
\cite{AtseriasLauriaNordstrom2016}), and vice-versa.

We will use the following lemma, which is a variant of Lemma~4.4
in~\cite{AtseriasLauriaNordstrom2016}:

\begin{lemma} \label{lem:ALN} Let~$w \geq 2$ be an integer, let~$C$ be
  a clause with at most~$w$ literals, let~$D$ be a clause with at
  most~$w-1$ literals, and let~$X$ be a variable that does not appear
  in~$D$.  Then the following four inequalities have Sherali-Adams
  proofs (from nothing) of constant monomial size and
  degree~$w$: \begin{enumerate} \itemsep=0pt
  \item~$ T(X \vee \overline{X}) \geq 0$,
  \item~$ - T(D \vee \overline{X}) - T(D \vee X) + T(D) \geq
    0$, \item~$ - T(D) + T(D \vee \overline{X}) + T(D \vee X) \geq
    0$, \item~$ - T(C) \geq 0$.
\end{enumerate}
\end{lemma}

\begin{proof}
Let $D = \bigvee_{i \in Y} X_i \vee \bigvee_{j \in Z} X_j$
and $C = \bigvee_{i \in Y'} X_i \vee \bigvee_{j \in Z'} X_j$.
Then
\begin{enumerate} \itemsep=0pt
\item $T(X \vee \overline{X}) = (1-X-\bar{X}) \cdot X + (X^2-X)$,
\item  $- T(D \vee \overline{X}) -T(D \vee X) + T(D) =  
  (X+\bar{X}-1) \cdot \prod_{i \in Y} \bar{X}_i \prod_{j \in Z} X_j$, 
\item $- T(D) + T(D \vee \overline{X}) + T(D \vee X) =
  (1-X-\bar{X}) \cdot \prod_{i \in Y} \bar{X}_i \prod_{j \in
  Z} X_j$,
\item $-T(C) = 1\cdot \prod_{i \in Y'} \bar{X}_i \prod_{j \in Z'} X_j$.
\end{enumerate}
The claim on the monomial size and the degree follows.
\end{proof}

Now we are ready to state and prove the first half of
Theorem~\ref{thm:circularresolutionvsSA}.

\begin{lemma}\label{lmm:circulartosar}
  Let~$A_1,\ldots,A_m$ and~$A$ be %
  clauses. If there is a Circular Resolution proof of~$A$
  from~$A_1,\ldots,A_m$ of length~$s$ and width~$w$, then there is a
  Sherali-Adams proof of~$T(A)\geq 0$
  from~$T(A_1) \geq 0,\ldots,T(A_m) \geq 0$ of monomial size~$O(s)$
  and degree~$w$.
\end{lemma}

\begin{proof}
  Let~$\Pi$ be a Circular Resolution proof of~$A$
  from~$A_1,\ldots,A_m$, and let~$F$ be the corresponding flow
  assignment.  Let~$I$ and~$J$ be the sets of inference- and
  formula-vertices of~$G(\Pi)$, and let~$B(u)$ denote the balance of
  formula-vertex~$u \in J$ in~$G(\Pi,F)$. For each
  formula-vertex~$u \in J$ labelled by formula~$A_u$, define the
  polynomial~$P_u := T(A_u)$. For each inference-vertex~$w \in I$
  labelled by rule~$R$, with sets of in- and out-neighbours~$N^-$
  and~$N^+$, respectively, define the
  polynomial \begin{center} \begin{tabular}{llll}~$P_w$ &~$:=$
      &~$T(A_a)$ & if~$R = \text{axiom}$ with~$N^+ = \{a\}$, \\~$P_w$
      &~$:=$ &~$-T(A_a) - T(A_b) + T(A_c)$ & if~$R = \text{cut}$
      with~$N^- = \{a,b\}$ and~$N^+ = \{c\}$, \\~$P_w$ &~$:=$
      &~$-T(A_a) + T(A_b) + T(A_c)$ & if~$R = \text{split}$
      with~$N^- = \{a\}$ and~$N^+ =
      \{b,c\}$. \end{tabular} \end{center} By double counting, the
  following polynomial identity holds: \begin{equation} \sum_{u \in J}
    B(u) P_u = \sum_{w \in I} F(w)
    P_w. \label{eqn:identity} \end{equation} By hypothesis~$G(\Pi,F)$
  has a sink~$s$ labelled by the derived clause~$A$.
  Since~$B(s) > 0$, equation~\eqref{eqn:identity} rewrites
  into \begin{equation*} \sum_{w \in I} \frac{F(w)}{B(s)} P_w +
    \sum_{u \in J\setminus\{s\}} -\frac{B(u)}{B(s)} P_u =
    P_s. \end{equation*} We claim that this identity is a legitimate
  Sherali-Adams proof of~$T(A) \geq 0$ from the
  inequalities~$T(A_1) \geq 0,\ldots,T(A_m) \geq 0$.
  First,~$P_s = T(A_s) = T(A)$, i.e., the right-hand side is
  correct. Second, each term~$(F(w)/B(s)) P_w$ for~$w \in I$ is a sum
  of legitimate terms of a Sherali-Adams proof by the definition
  of~$P_w$ and Parts~1,~2 and~3 of Lemma~\ref{lem:ALN}. Third, since
  each source~$u \in I$ of~$G(\Pi,F)$ has~$B(u) < 0$ and is labelled
  by a formula in~$A_1,\ldots,A_m$, the term~$(-B(u)/B(s)) P_u$ of a
  source~$u \in I$ is a positive multiple of~$T(A_u)$ and hence also a
  legitimate term of a Sherali-Adams proof
  from~$T(A_1) \geq 0,\ldots,T(A_m) \geq 0$.  And fourth, since each
  non-source~$u \in I$ of~$G(\Pi,F)$ has~$B(u) \geq 0$, each
  term~$(-B(u)/B(s)) P_u$ of a non-source~$u \in I$ is a sum of
  legitimate terms of a Sherali-Adams proof by the definition of~$P_u$
  and Part~4 of Lemma~\ref{lem:ALN}.  The monomial size and degree of
  this Sherali-Adams proof are as claimed, and the proof of the Lemma
  is complete.
\end{proof}

\subsection{Simulation of Sherali-Adams} \label{sec:simulationof}

In this section we prove the other half of
Theorem~\ref{thm:circularresolutionvsSA}. We use the notation
from Section~\ref{sec:simulationby}.

\begin{lemma} \label{lmm:sartocircular} Let~$A_1,\ldots,A_m$ and~$A$
  be non-tautological clauses. If there is a Sherali-Adams proof
  of~$T(A) \geq 0$ from~$T(A_1) \geq 0, \ldots,T(A_m) \geq 0$ of
  monomial size~$s$ and degree~$d$, then there is a Circular
  Resolution proof of~$A$ from~$A_1,\ldots,A_m$ of length~$O(s)$ and
  width~$d$.
\end{lemma}

\begin{proof}
  Fix a Sherali-Adams proof of~$T(A) \geq 0$
  from~$T(A_1) \geq 0, \ldots,T(A_m) \geq 0$, say
\begin{equation}
\sum_{j=1}^t Q_j P_j = T(A), \label{eqn:proof}
\end{equation}
where~$Q_j$ is a non-negative linear combination of monomials on the
variables~$X_1,\ldots,X_n$ and~$\bar{X}_1,\ldots,\bar{X}_n$, and~$P_j$
is a polynomial from among~$T(A_1),\ldots,T(A_m)$ or from among the
polynomials in the list~\eqref{eqn:allowedindefinition} from the
definition of Sherali-Adams in Section~\ref{sec:prelim}.

Our goal is to massage the proof~\eqref{eqn:proof} until it becomes a
Circular Resolution proof in disguise.  Towards this, as a first step,
we claim that~\eqref{eqn:proof} can be transformed into a \emph{normalized
proof} of the form
\begin{equation}
\sum_{j=1}^{t'} Q'_j P'_j = T(A) \label{eqn:proofnormalized}
\end{equation}
that has the following properties:
\begin{enumerate} \itemsep=0pt
\item each~$Q'_j$ is a positive multiple of a multilinear monomial,
  and~$Q'_jP'_j$ is multilinear,
\item each~$P'_j$ is a polynomial among~$T(A_1),\ldots,T(A_m)$, or among
  the polynomials in the set
\begin{equation}
\{ -X_i\bar{X}_i, 1-X_i-\bar{X}_i, X_i+\bar{X}_i-1 : i \in [n]\} \cup \{1\}.
\label{eqn:allowedtwo}
\end{equation}
\end{enumerate}
Comparing~\eqref{eqn:allowedtwo} with the original
list~\eqref{eqn:allowedindefinition} in the definition of
Sherali-Adams, note that we have replaced the polynomials~$X_i-X_i^2$
and~$X_i^2-X_i$ by~$-X_i\bar{X}_i$. Note also that, by splitting
the~$Q_j$'s into their terms, we may assume without loss of generality
that each~$Q_j$ in~\eqref{eqn:proof} is a positive multiple of a
monomial on the variables~$X_1,\ldots,X_n$
and~$\bar{X}_1,\ldots,\bar{X}_n$.

In order to prove the claim we rely on the well-known fact that each
real-valued function over the Boolean domain has a unique representation
as a multilinear polynomial:

\begin{fact} \label{fact:uniqueness} For every natural number~$N$ and
  every function~$f : \{0,1\}^{N} \rightarrow \mathbb{R}$ there is a
  unique multilinear polynomial~$P$ with~$N$ variables
  satisfying~$P(a_1,\ldots,a_N) = f(a_1,\ldots,a_{N})$ for
  every~$a_1,\ldots,a_N \in \{0,1\}$.
\end{fact}

With this fact in hand, it suffices to convert each~$Q_jP_j$ in the
left-hand side of~\eqref{eqn:proof} into a~$Q'_jP'_j$ of the required
form (or~$0$), and check that~$Q_jP_j$ and~$Q'_jP'_j$ are equivalent
over the 0-1 assignments to its variables (without relying on the
constraint that~$\bar{X}_i = 1-X_i$). The claim will follow from the
combination of Fact~\ref{fact:uniqueness} and the fact that~$T(A)$ is
multilinear since, by assumption,~$A$ is non-tautological.

We proceed to the conversion of each~$Q_jP_j$ into a~$Q'_jP'_j$ of the
required form. Recall that we assumed already, without loss of
generality, that each~$Q_j$ is a positive multiple of a monomial. The
multilinearization of a monomial~$Q_j$ is the monomial~$M(Q_j)$ that
results from replacing every factor~$Y^k$ with~$k \geq 2$ in~$Q_j$
by~$Y$. Obviously~$Q_j$ and~$M(Q_j)$ agree on 0-1 assignments, but
replacing each~$Q_j$ by~$M(Q_j)$ is not enough to guarantee the normal
form that we are after. We need to proceed by cases on~$P_j$.

If~$P_j$ is one of the polynomials among~$T(A_1),\ldots,T(A_m)$,
say~$T(A_i)$, then we let~$Q'_j$ be~$M(Q_j)$ with every variable that
appears in~$A_i$ deleted, and let~$P'_j$ be~$T(A_i)$ itself. It is
obvious that this works. If~$P_j$ is~$1-X_i-\bar{X}_i$, then we
proceed by cases on whether~$Q_j$ contains~$X_i$ or~$\bar{X}_i$ or
both or neither. If~$Q_j$ contains neither~$X_i$ nor~$\bar{X}_i$, then
the choice~$Q'_j = M(Q_j)$ and~$P'_j = P_j$ works.  If~$Q_j$
contains~$X_i$ or~$\bar{X}_i$, call it~$Y$, but not both, then the
choice~$Q'_j = M(Q_j)/Y$ and~$P'_j = -X_i\bar{X}_i$ works. If~$Q_j$
contains both~$X_i$ and~$\bar{X}_i$, then the
choice~$Q'_j = M(Q_j)/(X_i\bar{X}_i)$ and~$P'_j = -X_i\bar{X}_i$
works. If~$P_j$ is~$X_i+\bar{X}_i-1$, then again we proceed by cases
on whether~$Q_j$ contains~$X_i$ or~$\bar{X}_i$ or both or neither.
If~$Q_j$ contains neither~$X_i$ nor~$\bar{X}_i$, then the
choice~$Q'_j = M(Q_j)$ and~$P'_j = P_j$ works.  If~$Q_j$
contains~$X_i$ or~$\bar{X}_i$, call it~$Y$, but not both, then the
choice~$Q'_j = M(Q_j)\bar{Y}$ and~$P'_j = 1$ works. If~$Q_j$ contains
both~$X_i$ and~$\bar{X}_i$, then the choice~$Q'_j = M(Q_j)$
and~$P'_j = 1$ works. If~$P_j$ is the polynomial~$1$, then the
choice~$Q'_j = M(Q_j)$ and~$P'_j = 1$ works. Finally, if~$P_j$ is of
the form~$X_i - X_i^2$ or~$X_i^2 - X_i$, then we replace~$Q_jP_j$
by~$0$. Observe that in this case~$Q_jP_j$ is always~$0$ over 0-1
assignments, and the conversion is correct. This completes the proof
that~\eqref{eqn:proofnormalized} exists.

It remains to be seen that the normalized
proof~\eqref{eqn:proofnormalized} is a Circular Resolution proof in
disguise. For each~$j \in [t']$, let~$a_j$ and~$M_j$ be the positive
real and the multilinear monomial, respectively, such
that~$Q'_j = a_j M_j$. Let also~$C_j$ be the unique clause on
the variables~$X_1,\ldots,X_n$ such that~$T(C_j) = -M_j$. Let~$[t']$
be partitioned into five
sets~$I_0 \cup I_1 \cup I_2 \cup I_3 \cup I_4$ where
\begin{enumerate} \itemsep=0pt
\item $I_0$ is the set of $j \in [t']$ such that
$P'_j = T(A_{i_j})$ for some $i_j \in [m]$,
\item $I_1$ is the set of $j \in [t']$ such that 
$P'_j = -X_{i_j}\bar{X}_{i_j}$ for some $i_j \in [n]$,
\item $I_2$ is the set of $j \in [t']$ such that
$P'_j = 1-X_{i_j}-\bar{X}_{i_j}$ for some $i_j \in [n]$,
\item $I_3$ is the set of $j \in [t']$ such that
$P'_j = X_{i_j}+\bar{X}_{i_j}-1$ for some $i_j \in [n]$,
\item $I_4$ is the set of $j \in [t']$ such that
$P'_j = 1$.
\end{enumerate}
Note that the pair of sets~$I_0$ and~$I_1$ is disjoint because each
clause~$A_{i_j}$ is non-tautological by assumption. Since the rest of
pairs are clearly disjoint, we get a partition of~$[t']$.  Define new
polynomials~$P''_j$ as follows:
\begin{enumerate} \itemsep=0pt
\item[] $P''_j := T(C_j \vee A_{i_j})$ for $j \in I_0$,
\item[] $P''_j := T(C_j \vee \overline{X_{i_j}} \vee X_{i_j})$ for $j \in I_1$,
\item[] $P''_j := - T(C_j) + T(C_j \vee \overline{X_{i_j}}) + 
T(C_j \vee X_{i_j})$ for $j \in I_2$,
\item[] $P''_j := - T(C_j \vee \overline{X_{i_j}}) - T(C_j \vee X_{i_j}) + 
T(C_j)$ for $j \in I_3$,
\item[] $P''_j := T(C_j)$ for $j \in I_4$.
\end{enumerate}
With this notation,~\eqref{eqn:proofnormalized} rewrites into
\begin{equation}
\sum_{j \in I_0} a_j P''_j +
\sum_{j \in I_1} a_j P''_j +
\sum_{j \in I_2} a_j P''_j + 
\sum_{j \in I_3} a_j P''_j
=
T(A) + \sum_{j \in I_4} a_j P''_j.
\label{eqn:proofalmostcircular}
\end{equation}

Finally we are ready to construct the circular proof. We build it by
listing the inference-vertices with their associated flows, and then
we identify together all the formula-vertices that are labelled by the
same clause. 

Intuitively,~$I_0$'s are weakenings of hypothesis clauses,~$I_1$'s are
weakenings of axioms,~$I_2$'s are cuts,~$I_3$'s are splits,
and~$I_4$'s can be thought of as left-overs, i.e., clauses that were
produced but never used. Formally, each~$j \in I_0$ becomes a chain
of~$|C_j|$ many split vertices that starts at the hypothesis
clause~$A_{i_j}$ and produces its weakening~$C_j \vee A_{i_j}$; all
split vertices in this chain have flow~$a_j$. Each~$j \in I_1$ becomes
a sequence that starts at one axiom vertex that
produces~$X_{i_j} \vee \overline{X_{i_j}}$ with flow~$a_j$, followed
by a chain of~$|C_j|$ many split vertices that produces its
weakening~$C_j \vee X_{i_j} \vee \overline{X_{i_j}}$; all split
vertices in this chain also have flow~$a_j$. Each~$j \in I_2$ becomes
one cut vertex that produces~$C_j$ from~$C_j \vee X_{i_j}$
and~$C_j \vee \overline{X_{i_j}}$ with flow~$a_j$. And
each~$j \in I_3$ becomes one split vertex that
produces~$C_j \vee X_{i_j}$ and~$C_j \vee \overline{X_{i_j}}$
from~$C_j$ with flow~$a_j$. Note that some of the produced clauses,
such as the second conclusion in certain intermediate splits, may
never be used as hypothesis in other rules.

This defines the inference-vertices of the proof graph. The
construction is completed by introducing one formula-vertex for each
different clause that occurs as a premise or a conclusion of
these inference-vertices. The construction was designed in such a way
that equation~\eqref{eqn:proofalmostcircular} witnesses that this
proof graph and its associated flow assignment makes a correct
circular proof.  In order to see this we need to compute the balances
of all the formula-vertices, check that some sink is labelled~$A$, and
check that all sources are labelled by formulas in~$A_1,\ldots,A_m$.

First note that, for each formula-vertex~$u$ whose labelling formula
appears as an intermediate formula in a chain of split inferences for
some~$j \in I_0 \cup I_1$, the contribution of this occurrence to the
balance of~$u$ is zero. Indeed, for being an intermediate formula in
the chain, the flow~$a_j$ of the inference that produces it cancels
out the flow~$a_j$ of the inference that consumes it. Second note that
the only formula-vertices~$u$ whose labelling formulas are of the
form~$\overline{X_{i_j}} \vee X_{i_j}$ are those that are introduced
by an axiom inference in case~$j \in I_1$; this follows from the
multilinearity of the~$Q'_jP'_j$ and the assumption that the~$A_i$ are
non-tautological.  It follows that the balance of such~$u$ is
positive; indeed its balance is the flow~$a_j$ of the axiom inference
that produces it. Third note that we do not need to compute the
balance of the formula-vertices~$u$ whose labelling formula is a
hypothesis among~$A_1,\ldots,A_m$; any balance is allowed for it. With
these three notes in mind it suffices to check that:
\begin{enumerate} \itemsep=0pt
\item some sink is labelled~$A$,
\item for~$j \in I_0$ the formula-vertex for~$C_j \vee A_{i_j}$ has
  non-negative balance,
\item for~$j \in I_1$ the formula-vertex
  for~$C_j \vee \overline{X_{i_j}} \vee X_{i_j}$ has non-negative
  balance,
\item for~$j \in I_2 \cup I_3$, the vertices
  for~$C_j$,~$C_j \vee \overline{X_{i_j}}$, and~$C_j \vee X_{i_j}$
  have non-negative balance.
\end{enumerate}
For (1) it suffices to note that~$T(A)$ appears in the right-hand side
of~\eqref{eqn:proofalmostcircular} with a positive multiplier, hence
its occurrences in the left-hand side must have multipliers that add
up to a positive quantity. For (2), (3), and (4) it suffices to note
that if~$B$ is one of the indicated formulas, then~$T(B)$ appears in
the left-hand side of~\eqref{eqn:proofalmostcircular} with multipliers
that add-up to either a positive quantity or~$0$ depending on
whether~$T(B)$ appears in the right-hand side or not,
respectively. Hence its balance is non-negative, and the proof that
the circular proof is correct is complete. The claim that the length
of this proof is~$O(s)$ and its width is~$d$ follows by inspection.
\end{proof}

\section{Circular Frege vs Tree-Like Frege}

For some weak proof systems, such as Resolution, it makes a great deal
of difference whether the proof-graph has tree-like structure or not
\cite{BonetEstebanGalesiJohannsen2000}. For stronger proofs systems,
such as Frege, this is not the case. Indeed Tree-like Frege
polynomially simulates Dag-like Frege, and this holds true of any
inference-based proof system with the set of all formulas as its set
of allowed formulas, and a finite set of inference rules that is
implicationally complete \cite{Kra94BoundedArithmetic}.  Since
circular proofs further generalize the structure of the proof-graph,
it is interesting to discuss whether circular proofs in Frege are
complexity-wise more powerful than standard Frege proofs.

It turns out that this is not the case. 
In this section we show how to efficiently simulate Circular Frege, as
defined in Section~\ref{sec:definition}, by standard Frege proofs.

\begin{theorem}
Tree-like Frege and Circular Frege polynomially simulate each other.
\end{theorem}

The main idea underlying the simulation of Circular Frege by standard
Frege is to formalize, in standard Dag-like Frege itself, the LP-based
proof of soundness of Frege circular proofs; cf.\ the \emph{second
  proof} of Theorem~\ref{thm:soundness}.
To do that we use
a formalization of linear arithmetic in Frege, due to
Buss~\cite{Buss1987} and Goerdt~\cite{Goerdt1991CPvsFrege}. This
formalization was originally designed to simulate counting arguments
and Cutting Planes in Frege. Since Cutting Planes subsumes
LP-reasoning, the core of the LP-based proof of
Theorem~\ref{thm:soundness} can be formalized in it.  It should be
pointed out that, while the Frege systems that were used by
Buss~\cite{Buss1987} and Goerdt~\cite{Goerdt1991CPvsFrege} are not
exactly the same as ours, their results apply to our definition of
Frege; this follows from the well-known robustness properties of Frege
systems which guarantee that any two (Dag-like) Frege systems
polynomially simulate each other~\cite{CookReckhow1979}.

\subsection{Formalization of Linear Arithmetic in Frege}

We collect the relevant parts of Goerdt's results in a single
theorem. Before that, we need to introduce some notation.  Fix a set
of~$n$ Boolean variables~$X_1,\ldots,X_n$.  Let~$\mathscr{L}$ denote
the collection of all linear inequalities of the form
\begin{equation}
a_1 X_1 + \cdots + a_n X_n \geq b,
\label{eqn:inequality}
\end{equation} 
where~$X_1,\ldots,X_n$ are formal variables and~$a_1,\ldots,a_n$
and~$b$ are integers. 
Our notation includes inequalities of the form~$0 \ge b$,
and~$a_{i_1} X_{i_1} + \cdots + a_{i_t} X_{i_t} \ge b$ by deleting the
terms with zero coefficients. By the \emph{set of variables} of an
inequality we mean the set of variables that appear in it with
non-zero coefficient.

Let~$\ell$ and~$\ell'$ denote two inequalities in~$\mathscr{L}$,
with sequences of coefficients~$b,a_1,\ldots,a_n$
and~$b',a'_1,\ldots,a'_n$.
Let~$c$ be a positive integer.  We
write~$c \cdot \ell$ and~$\ell + \ell'$ for the following two
inequalities, respectively:
\begin{eqnarray*}
& ca_1 X_1 + \cdots + ca_n X_n \geq cb, \\
& (a_1 + a'_1) X_1 + \cdots + (a_n + a'_n) X_n \geq (b + b').
\end{eqnarray*}
Let~$\mathscr{F}$ denote the collection of all propositional formulas
in negation normal form. We move back and forth between the truth
assignments and the~$0$-$1$ assignments for the same sets of
variables.  If~$f : \{X_1,\ldots,X_n\} \rightarrow \{0,1\}$ denotes
such an assignment, and~$\ell$ and~$A$ denote, respectively, an
inequality and a formula on the variables~$X_1,\ldots,X_n$, then we
write~$f(\ell)$ and~$f(A)$ for their truth values under~$f$.
Concretely, if~$\ell$ is as in~\eqref{eqn:inequality}, then~$f(\ell)$
is true if and only if~$a_1 f(X_1) + \cdots + a_n f(X_n)$ is at
least~$b$.

For concreteness, we specify a size measure for inequalities. An
inequality~$\ell$ as in~\eqref{eqn:inequality} is represented by the
sequence of the binary encodings of its
coefficients~$b,a_1,\ldots,a_n$; its size
is~$\Theta(n+\log_2(|b|) + \sum_{i=1}^n \log_2(|a_i|))$, with the
convention that~$\log_2(0) = 0$, and where~$|v|$ denotes the absolute
value of~$v \in \integers$.

\begin{theorem}[\cite{Goerdt1991CPvsFrege}]
  \label{thm:cpencoding}
  There is a mapping $I : \mathscr{L} \rightarrow \mathscr{F}$ that
  takes linear inequalities to formulas and that has the following
  properties.  For every two inequalities $\ell$ and $\ell'$ in
  $\mathscr{L}$, every positive integer $c$, every truth assignment
  $f$, and every variable $X$, the following hold:
  \begin{enumerate} \itemsep=0pt
  \item $I(\ell)$ has size polynomial in the size of $\ell$,
  \item~$I(\ell)$ has the same variables as~$\ell$,
    and~$f(\ell) = f(I(\ell))$,
  \item there is a polynomial-size Frege proof of $I(\ell+\ell')$ from
  $I(\ell)$ and $I(\ell')$,
  \item there is a polynomial-size Frege proof of $I(c\cdot\ell)$ from
  $I(\ell)$,
  \item there is a polynomial-size Frege proof of $I(\ell)$ from
  $I(c\cdot \ell)$,
  \item there is a polynomial-size Frege proof of $I(X \geq 1)$ from $X$,
  \item there is a polynomial-size Frege proof of $I(-X \geq 0)$ from $\overline{X}$,
  \item there is a polynomial-size Frege proof of $I(-X \geq -1)$ from nothing,
  \item there is a polynomial-size Frege proof of $I(X \geq 0)$ from nothing,
  \item there is a polynomial-size Frege proof of $I(0 \geq 0)$ from nothing,
  \item there is a polynomial-size Frege proof of $\emptyformula$ from
    $I(0 \geq 1)$.
  \end{enumerate}
  Moreover, the mapping~$I$ and the Frege proofs in Points~3--11 are
  all computable in time that is bounded by a fixed polynomial in the
  sizes of the input and the output inequalities.
\end{theorem}

\begin{proof}
  All this can be found in Goerdt's
  article~\cite{Goerdt1991CPvsFrege}, which in turn builds on Buss's
  seminal~\cite{Buss1987}: the definition of the mapping~$I$ is in
  Section 2.6 of Goerdt's article, Points 1--2, and Points 6--11
  follow by inspection of the definition of~$I$ given there, and
  Points 3, 4, and 5 are Theorems~3.1,~3.5, and~3.6 in Goerdt's
  article, respectively.
\end{proof}

Technically the inequalities are defined in~\eqref{eqn:inequality}
with the variables on the left side of the~$\geq$ relation symbol and
the constant on the right side. For readability reasons, in the
following we apply~$I$ to linear inequalities written in any form that
follows from adding or subtracting the same linear terms from both
sides of an inequality in official form; e.g., when we
write~$I(-1 \geq 0)$ or~$I(X_1 - 2 \geq X_2)$ we really
mean~$I(0 \geq 1)$ and~$I(X_1-X_2\geq 2)$, respectively.

\subsection{Proof of the Simulation}

This section is devoted to the proof of Theorem 4. The statement that
Circular Frege polynomially simulates Tree-like Frege follows from the
discussion in Section~\ref{sec:symmetricrules}. We concentrate on the
reverse simulation. Since it is known that Tree-like Frege
polynomially simulates Dag-like Frege, it suffices to do the
simulation through dag-like proofs. Also we claim that it suffices to
do the simulation only for refutations. Indeed, from a short circular
proof of~$A$ from~$\mathscr{H}$ we can get a short circular refutation
of~$\mathscr{H} \cup \{\overline{A}\}$ by adding a cut between the
derived~$A$ and the new hypothesis~$\overline{A}$, with flow equal to
the balance of the formula-vertex of~$A$. And from a short dag-like
refutation of~$\mathscr{H} \cup \{\overline{A}\}$ we can get a short
dag-like proof of~$A$ from~$\mathscr{H}$ by replacing each use of the
hypothesis formula~$\overline{A}$ by the axiom
instance~$A \vee \overline{A}$.

Let~$\Pi$ be a Circular Frege refutation of a set of hypothesis
formulas~$\mathscr{H}$. Let us choose an arbitrary~$s$ among the
formula-vertices in~$G(\Pi)$ that are labelled by the empty formula.
The simulation goes in three steps. In the first step we build an
infeasible linear program~$P = \{\ell_1,\ldots,\ell_m\}$ that has one
variable~$Z_u$ for each formula-vertex~$u$ of~$G(\Pi)$. Any witness of
infeasibility of~$P$ will witness the soundness of~$\Pi$ as a circular
refutation. This is done by closely following the second proof of
soundness of circular proofs; cf.~Theorem~\ref{thm:soundness}. In the
second step we apply Theorem~\ref{thm:cpencoding} to convert an
LP-based witness of infeasibility for~$P$ into a Frege refutation of
the set of
formulas~$\mathscr{H}' := \{ I(\ell_1),\ldots,I(\ell_m) \}$. Here~$I$
is the mapping from Theorem~\ref{thm:cpencoding}. In the third step we
apply the substitution defined by~$Z_u := A_u$ to this Frege
refutation, where~$A_u$ is the formula that labels the
formula-vertex~$u$, and we apply Theorem~\ref{thm:cpencoding} again to
show that each formula in the substituted~$\mathscr{H'}$ has an
efficient Frege proof from~$\mathscr{H}$.

\bigskip
\noindent\textit{First step}.\;\;
The linear program~$P$ has one variable~$Z_u$ for each
formula-vertex~$u \in J$ in~$G(\Pi)$, and two sets of
inequalities~$P_J = \{ \ell_u : u \in J \}$
and~$P_I = \{ \ell_w : w \in I \}$ indexed by the sets of
formula-vertices~$J$ and inference-vertices~$I$ of~$G(\Pi)$,
respectively. Concretely, for each formula-vertex~$u \in J$ the
inequality~$\ell_u$ is defined as follows:
\begin{center}
  \begin{tabular}{rcll}
    $-Z_s$ & $\geq$ & $0$ & for the formula-vertex $s$
                           labelled by the derived empty formula, \\
    $-(1-Z_u)$ & $\geq$ & $0$ & if $u$ is a formula-vertex of
                               a hypothesis formula, \\
    $(1-Z_u)$ & $\geq$ & $0$ & if $u$ is any other formula-vertex.
  \end{tabular}
\end{center}
For each inference-vertex~$w$ with labelling inference rule~$R$ and
in- and out-neighbours~$N^-$ and~$N^+$, respectively, the
inequality~$\ell_w$ is defined as follows:
\begin{center}
  \begin{tabular}{rcll}
    $-(1-Z_a)$ & $\geq$ & $0$ & if $R = \text{axiom}$ with $N^+ =
                               \{a\}$, \\
    $(1-Z_a) + (1-Z_b) - (1-Z_c)$ & $\geq$ & $0$ & if $R = \text{cut}$
                                                with $N^- = \{a,b\}$
                                                and $N^+ = \{c\}$, \\
    $(1-Z_a) - (1-Z_b) - (1-Z_c)$ & $\geq$ & $0$ & if $R = \text{split}$ with $N^- = \{a\}$ and $N^+
                                                = \{b,c\}$.
  \end{tabular}
\end{center}
A certificate of the infeasibility of~$P = P_I \cup P_J$ is given by
two assignments of non-negative weights~$(b_u : u \in J)$
and~$(c_w : w \in I)$ for the inequalities in~$P_J$ and~$P_I$,
respectively, in such a way that the corresponding positive linear
combination
\begin{equation}
  \sum_{u \in J} b_u \cdot \ell_u + \sum_{w \in I} c_w \cdot \ell_w \label{eqn:plc}
\end{equation}
simplifies to the trivially false inequality~$-1 \geq 0$ (or
equivalently~$0 \geq 1$).
In turn, such an assignment of weights can be shown to exist from the
assumption that~$\Pi$ is a valid circular refutation: let~$F$ be a
flow assignment that witnesses that~$\Pi$ is a valid proof and~$B$ the
corresponding balance on the formula-vertices, and
set~$b_u := -B(u)/B(s)$ for each formula-vertex~$u \in J$ that is a
source of~$G(\Pi,F)$, set~$b_u := B(u)/B(s)$ for each
formula-vertex~$u \in J$ that is not a source of~$G(\Pi,F)$, and
set~$c_w := F(w)/B(s)$ for every inference-vertex~$w \in I$. Note
that~$B(s)$ is strictly positive because~$s$ must be a sink
of~$G(\Pi,F)$.  This means that each~$b_u$ is well-defined and
non-negative because the balance of all formula-vertices except the
sources is non-negative in~$G(\Pi,F)$.
To prove that this assignment of weights makes~\eqref{eqn:plc} to
simplify to $-1 \geq 0$ observe that the second sum in \eqref{eqn:plc}
is equal to the right hand side of \eqref{eqn:second} divided by
$B(s)$. Hence \eqref{eqn:plc} can be written as
\begin{equation}
  -Z_{s} + \sum_{u \in J/\{s\}} \frac{B(u)}{B(s)} \cdot (1-Z_{u})
  \quad
  +
  \quad
  (Z_{s}- 1 ) + \sum_{u \in J/\{s\}} -\frac{B(u)}{B(s)} (1-Z_u) \geq 0\;.
\end{equation}
This completes the first step of the simulation.

\bigskip
\noindent\textit{Second step}.\;\; The second step is a direct
application of Theorem~\ref{thm:cpencoding}: Define the non-negative
integers~$b'_u = b_u \cdot B(s)$ and~$c'_w = c_w \cdot B(s)$. Start
at~$\mathscr{H}' = \{I(\ell_u) : u \in J\}\cup \{I(\ell_w) : w \in
I\}$.  By Points~4 and~10 in Theorem~\ref{thm:cpencoding}, obtain
Frege proofs of~$I(b'_u \cdot \ell_u)$ and~$I(c_w' \cdot \ell_w)$ for
each~$u \in J$ and each~$w \in I$. Now let~$\ell'$ denote the positive
linear combination defined as in~\eqref{eqn:plc} with~$b_u$ and~$c_w$
replaced by~$b'_u$ and~$c'_w$, respectively. Recall that~$\ell$
is~$-1 \geq 0$ and hence~$\ell'$ is~$-B(s) \geq 0$. By Point~3, obtain
a Frege proof of~$I(-B(s) \geq 0)$. By Point~5, obtain a Frege proof
of~$I(-1 \geq 0)$, or equivalently~$I(0 \geq 1)$.  Finally, by
Point~11, obtain the Frege proof of $\emptyformula$.

\bigskip
\noindent\textit{Third step}.\;\;
We start the third step by applying the substitution defined
by~$Z_u := A_u$ to the refutation of~$\mathscr{H}'$, where~$A_u$ is
again the formula that labels the formula-vertex~$u$. For
each~$v \in I \cup J$, let~$I(\ell_v)^*$ denote the result of applying
this substitution to~$I(\ell_v)$. To complete the step we need to get
polynomial-size Frege proofs of~$I(\ell_v)^*$ from~$\mathscr{H}$, for
each~$v \in I \cup J$. We do this as a \emph{less direct} application
of Theorem~\ref{thm:cpencoding}.

For each formula-vertex~$u \in J$ of a hypothesis formula
in~$\mathscr{H}$, we get a Frege proof of~$I(\ell_u)^*$ from~$A_u$ by
applying the substitution~$X := A_u$ to the Frege proof given by
Point~6 in Theorem~\ref{thm:cpencoding}. When~$u$ is the
formula-vertex of the derived empty formula, we get a Frege proof
of~$I(\ell_u)^*$ from~$\overline{0}$ by applying the
substitution~$X := 0$ to the Frege proof given by Point~7 in
Theorem~\ref{thm:cpencoding}. Since~$\overline{0}$ is the conclusion
of an instance of the axiom rule of Frege
(namely~$0 \vee \overline{0}$), this is a Frege proof of~$I(\ell_u)^*$
from nothing. For every other formula-vertex~$u$, we get a Frege proof
of~$I(\ell_u)^*$ from nothing by applying the substitution~$X := A_u$
to the Frege proof given by Point~8 in Theorem~\ref{thm:cpencoding}.

For each inference-vertex~$w \in I$, with labelling rule~$R$ and in-
and out-neighbours~$N^-$ and~$N^+$, we proceed as follows.  By Point~2
in Theorem~\ref{thm:cpencoding} and the soundness of~$R$, first note
that~$I(\ell_w)^*$ is a propositional tautology. We claim that, in
addition, this tautology is obtained by applying a substitution to
another tautology~$T$ that has at most two propositional variables~$X$
and~$Y$.  Concretely,~$T$ will itself be the result of applying a
substitution to~$I(\ell_w)$. We define~$T$ by cases depending on what
rule~$R$ is. If~$R$ is the axiom rule and~$N^+ = \{a\}$, then we
take~$T$ to be the result of applying the
substituion~$Z_a := X \vee \overline{X}$ to~$I(\ell_w)$.  If~$R$ is
the cut rule,~$N^- = \{a,b\}$ and~$N^+ = \{c\}$, then we take~$T$ to
be the result of applying the substitution~$Z_a := Y \vee
X$,~$Z_b := Y \vee \overline{X}$, and~$Z_c := Y$
to~$I(\ell_w)$. If~$R$ is the split rule,~$N^- = \{a\}$
and~$N^+ = \{b,c\}$, then we take~$T$ to be the result of applying the
substitution~$Z_a := Y$,~$Z_b := Y \vee X$,
and~$Z_c := Y \vee \overline{X}$ to~$I(\ell_w)$. By Point~2 in
Theorem~\ref{thm:cpencoding} and the soundness of~$R$, in all three
cases~$T$ is a tautology with at most two propositional variables.  By
the completeness of Frege,~$T$ has a constant-size Frege
proof. Applying the substitution that turns~$T$ into~$I(\ell_w)^*$ to
this proof we get a polynomial-size Frege proof of~$I(\ell_w)^*$ as
desired. This completes the third step, and the proof.

\commentout{  

\subsection{Old}

We will also need the following property of Frege proofs.
\begin{theorem}[Theorem 4.4.9 in~\cite{Kra94BoundedArithmetic}]
  \label{thm:frege}
  Consider a Frege a proof $\Pi$ of a formula $F$ from a set of
  formulas $\mathcal{H}$, all over propositional variables
  $x_{1}, \ldots, x_{n}$.
  Define $\Pi'$, $F'$, and $\mathcal{H}'$ by applying the substitution
  $x_{i} \leftarrow B_{i}$ to $\Pi$, $F$, and $\mathcal{H}$, respectively,
  where $B_{1}, \ldots, B_{n}$ are formulas over propositional
  variables $y_{1}, \ldots, y_{m}$.
  Then $\Pi'$ is a proof, in the same Frege proof system, of formula
  $F'$ from a set of formulas $\mathcal{H}'$, all over propositional variables
  $y_{1}, \ldots, y{n}$.
\end{theorem}

\begin{theorem}
Any proof in circular Frege can be transformed in polynomial time into
a proof in tree-like Frege. 
\end{theorem}
\begin{proof}
  Consider a circular Frege proof $\Pi$ that derives a formula $C$
  from a set of formulas $\mathcal{H}$.
  Because of implicational completeness of tree-like Frege we can
  assume that the goal formula is actually the false formula $\bot$, and
  without loss of generality we assume that $\bot$ is not in the set
  of hypothesis.
  Since it is well known that tree-like and dag-like Frege are
  equivalent~\cite[Lemma 4.4.8]{Kra94BoundedArithmetic}, we just
  convert the circular proof into dag-like Frege.
  The simulation goes in three steps
  \begin{enumerate}
    \item we build a linear program that proves the soundness of the
    circular refutation;
    \item we use the Frege encoding of linear programming
    reasoning by Goerdt~\cite{Goerdt1991CPvsFrege};
    \item we show that the encoding of the initial inequalities are
    efficiently provable from $\mathcal{H}$ in Frege.
  \end{enumerate}
  
  The first step of the simulation goes through the construction of an
  unsatisfiable linear. For every formula $F$ occurring in a formula
  vertex of $\Pi$ we consider a variable $z_{F}$ with the intended
  meaning that
  \begin{equation*}
    z_{F}=
    \begin{cases}
      1 & \text{if $F$ is satisfied;} \\
      0 & \text{otherwise.}
    \end{cases}
  \end{equation*}  
  We furthermore define the following linear inequalities on these
  variables
  \begin{align}
    &\label{eq:linvar} 0 \leq z_{F} \leq 1
    & \text{for every $F$ labeling a formula vertex;}
    \\
    &\label{eq:linfalse} z_{\bot} \leq 0
    & \text{the false formula;}
    \\
    &\label{eq:linaxiom} z_{F} \geq 1
    & \text{for each $F \in \mathcal{H}$ in the proof;}
    \\
    &\label{eq:lincutrule}
      (1- z_{F \vee \bar{A}}) + (1 - z_{F \vee A}) - n(1 - z_{F}) \geq 0
    & \text{for each cut rule used in the proof;}
    \\
    &\label{eq:linsplitrule}
      (1 - z_{F}) - (1- z_{F \vee \bar{A}}) - (1 - z_{F \vee A}) \geq 0
    & \text{for each split rule used in the proof.}
  \end{align}
  The inequalities~\eqref{eq:lincutrule} and~\eqref{eq:linsplitrule}
  enforce that if the variables corresponding to the premises of
  a rule are all $1$, then the variables corresponding to the
  conclusions have to be $1$ as well.
  By summing all inequalities~\eqref{eq:lincutrule} and
  \eqref{eq:linsplitrule} and we get
  \begin{equation}\label{eq:linsum}
    \sum_{F} n_{F} (1 - z_{F}) \geq 0
  \end{equation}
  where the sum is over all the formulas occurring in the proofs,
  where $ - n_{F}$ is equal to the sum of the balances of all formula
  vertices labeled by $F$.
  Notice that $n_{F} > 0$ only if $F \in \mathcal{H}$, therefore we can
  add positive multiples of rule~\eqref{eq:linaxiom} so that the sum
  (\ref{eq:linsum}) would only contain summands with $ n_{F} < 0$
  (which are the goal formulas of the proof).
  Recall that $\bot$ is among the goal formulas. By summing also
  positive multiples of the inequalities in (\ref{eq:linvar}) and dividing by
  a positive integer we get inequality
  \begin{equation*}
     - (1 - z_{\bot}) \geq 0
  \end{equation*}
  which, together with $z_{\bot} = 1$ gives the contradiction
  $0 \geq 1$.
  
  For the second step we use the encoding from
  Theorem~\ref{thm:cpencoding}, and we simulate efficiently the
  refutation of the linear program in Frege. This is possible because
  the contradiction is deduced using only rules in the theorem.
  
  So far we have only obtained a refutation defined on the
  propositional variables $z$, starting from the set of formulas
  $\{I_{\ell}\}_{\ell}$ where $\ell$ is among
  \begin{itemize}
    \item circular Frege rules as in (\ref{eq:lincutrule}) and (\ref{eq:linsplitrule});
    \item inequalities as in (\ref{eq:linvar}) and (\ref{eq:linaxiom});
    \item $z_{\bot} \leq 0$.
  \end{itemize}
  In order to turn that into a refutation from $\mathcal{H}$ we substitute each
  variable $z_{F}$ with the corresponding formula $F$. Frege proofs
  are closed under substitution of variables by formulas, and the
  resulting proof is only polynomially larger, since every $F$
  labeling a formula vertex has polynomial size.
  The new proof after substitution is only larger by
  a polynomial factor (see Theorem~\ref{thm:frege}).
  
  The proof is almost a full simulation (in dag-like Frege) of the
  circular Frege refutation, but the premises are not yet formulas
  from $\mathcal{H}$ and we need still to show how to derive them.
  Each of the initial formulas was obtained as follow: we started with
  a linear inequality $\ell\!\left(\vec{z}\right)$ over the $z$
  variables; then we encoded this inequality as a propositional
  formula $I\!\left(\vec{z}\right)$; then we substituted each variable
  $z_{F}$ with the corresponding formula $F$.

  The case of the inequality $\ell\!\left(\vec{z}\right)$ being either
  $ z_{\bot} \leq 0 $ or from (\ref{eq:linvar}) is easy. The resulting
  formula after the substitution is a tautology with a constant number
  of symbols. Therefore there is a Frege derivation of constant size
  for it.

  The most tricky case is when inequality $\ell\!\left(\vec{z}\right)$
  comes from an instantiation of a rule, as
  in~(\ref{eq:lincutrule,eq:linsplitrule}).
  Observe that $\ell\!\left(\vec{z}\right)$ has just a constant number
  of $z$ variables, and its corresponding encoding $I$ likewise. Let's
  call $z_{F_{1}}, \ldots, z_{F_{3}}$ such variables, corresponding to
  formulas $F_{1}, \ldots F_{3}$.
  Notice that $I(F_{1},F_{2},F_{3})$ is a tautology because the
  Frege rule is sound. Furthermore, each $F_{i}$ is obtained by
  substitution from some formulas $B_{i}$ that have a constant number
  of atoms. Propositional formula $I(B_{1},B_{2},B_{3})$ is a tautology too,
  because it is just a formula that claims the soundness of
  a corresponding circular Frege rule.
  Therefore it has a constant size proof $\Pi$ in Frege. The final step is to
  notice that by applying the same substitution that maps $B_{i}$ to
  $F_{i}$ to the whole proof $\Pi$ we get a small Frege proof $\Pi'$ of
  $I(F_{1},F_{2},F_{3})$, which is want we needed.
\end{proof}
}

\section{Concluding Remarks} \label{sec:concludingremarks}

One immediate consequence of Theorem~\ref{thm:circularresolutionvsSA} is that
there is a polynomial-time algorithm that automates the search for
Circular Resolution proofs of bounded width:

\begin{corollary} \label{cor:automate} There is an algorithm that,
  given an integer parameter $w$ and a set of clauses $A_1,\ldots,A_m$
  and $A$ with $n$ variables, returns a Circular Resolution proof of
  width $w$ of clause $A$ from $A_1,\ldots,A_m$, if there is one, and
  the algorithm runs in time polynomial in $m$ and $n^w$.
\end{corollary}

The proof-search algorithm of Corollary~\ref{cor:automate} relies on
linear programming because it relies on our translations to and from
Sherali-Adams, whose automating algorithm does rely on linear
programming~\cite{SheraliAdams1990}.  Based on the fact that the
number of clauses of width~$w$ is about $n^w$, a direct proof of
Corollary~\ref{cor:automate} is also possible, but as far as we see it
still relies on linear programming for finding the flow assignment.
It remains as an open problem whether a more combinatorial algorithm
exists for the same task.

Another consequence of the equivalence with Sherali-Adams is that
Circular Resolution has a length-width relationship in the style of
the one for Dag-like
Resolution~\cite{BenSassonWigderson2001}. This follows from
Theorem~\ref{thm:circularresolutionvsSA} in combination with the
size-degree relationship that is known to hold for Sherali-Adams (see
\cite{SegerlindPitassi2012,AtseriasOberwolfachSlides2017}). 
Combining this with the known lower bounds for Sherali-Adams
(see \cite{Grigoriev2001,SegerlindPitassi2012}), we get the following:

\begin{corollary}
  There are families of 3-CNF formulas $(F_n)_{n \geq 1}$, where $F_n$
  has $O(n)$ variables and $O(n)$ clauses, such that every Circular
  Resolution refutation of $F_n$ has width $\Omega(n)$ and
  size~$2^{\Omega(n)}$.
\end{corollary}

It should be noticed that, unlike the well-known observation that
tree-like and dag-like width are equivalent measures for Resolution,
these are \emph{not} equivalent to width in Circular Resolution.
The sparse graph pigeonhole principle from
Section~\ref{sec:resolution} illustrates the point. This shows that
Circular Resolution proofs of bounded-width proofs cannot be
\emph{unfolded} into bounded-width tree-like Resolution proofs.

This observation also explains, perhaps, why our proof that Circular
Frege simulates Tree-like Frege goes via a very indirect translation.
It also raises one further question (and answer).  It is known that
Tree-like Bounded-Depth Frege simulates Dag-like Bounded-Depth Frege,
at the cost of increasing the depth by one. Could the simulation of
Circular Frege by Tree-like Frege be made to preserve bounded depth?
The (negative) answer is also provided by the pigeonhole principle
which is known to be hard for Bounded-Depth Frege
\cite{Ajtai1988,PitassiBeameImpagliazzo1993,KrajicekPudlakWoods1995},
but is easy for Circular Resolution, and hence for
Circular~Depth-1~Frege.

One last aspect of the equivalence between Circular Resolution and
Sherali-Adams concerns the theory of SAT-solving. As is well-known,
state-of-the-art SAT-solvers produce Resolution proofs as certificates
of unsatisfiability and, as a result, will not be able to handle
counting arguments of pigeonhole type. This has motivated the study of
so-called \emph{pseudo-Boolean solvers} that handle counting
constraints and reasoning through specialized syntax and inference
rules. The equivalence between Circular Resolution and Sherali-Adams
suggests a completely different approach to incorporate counting
capabilities: instead of enhancing the syntax, keep it to clauses but
enhance the \emph{proof-shapes}. Whether circular proof-shapes can be
handled in a sufficiently effective and efficient way is of course in
doubt, but certainly a question worth studying.

It turns out that Circular Resolution has unexpected connections with
Dual Rail MaxSAT Resolution~\cite{IgnatievEtAl2017TacklingLimits}.
MaxSAT Resolution is a variant of resolution where proofs give upper
bounds on the number of clauses of the CNF that can be satisfied
simultaneously. At the very least, when the upper bound is less than
the number of clauses, MaxSAT resolution provides a refutation of
the formula.
The Dual Rail encoding is a special encoding of CNF formulas, and Dual
Rail MaxSAT Resolution is defined to be MaxSAT resolution applied to the
Dual Rail encoding of the input formula.
It is argued in \cite{BonetEtAl2018MaxsatResolution} that Dual
Rail encoding gives strength to the proof system, providing a polynomial
refutation of the pigeonhole principle formula.
Following the conference version of this paper
\cite{AtseriasLauria2019}, it was argued in~\cite{Vinyals2019DualRail}
that Circular Resolution polynomially simulates Dual Rail MaxSAT
Resolution, in the sense that when the Dual Rail encoding of a CNF
formula~$F$ has a MaxSAT Resolution refutation of length~$\ell$ and
width~$w$, then~$F$ has a Circular Resolution refutation of
length~$O(\ell w)$.  This is interesting per se and provides yet
another proof of Corollary~\ref{cor:upperbound}.  The exact relative
strength of Circular Resolution, Dual Rail MaxSAT, and other systems
for MaxSAT is studied further in \cite{BonetLevy2020,
  LarrosaRollon2020}.

\bigskip
\noindent\textbf{Acknowledgments}\; Both authors were partially funded
by European Research Council (ERC) under the European Union's Horizon
2020 research and innovation programme, grant agreement ERC-2014-CoG
648276 (AUTAR). First author partially funded by MINECO through
TIN2013-48031-C4-1-P (TASSAT2). We acknowledge the work of Jordi Coll
who conducted experimental results for finding and visualizing actual
circular resolution proofs of small instances of the sparse pigeonhole
principle. We thank Moritz M\"uller for comments and for carefully
reading an earlier version of the paper.

\bibliographystyle{plain}

\end{document}